\documentclass[preprint,12pt]{elsarticle}

\usepackage{amsmath,amsthm,epsfig,amssymb,verbatim,amsopn,subfigure,multirow,hyperref}
\usepackage{balance,nopageno}
\usepackage{mathtools}
\usepackage[usenames,dvipsnames]{color}
\usepackage[all]{xy}  
\usepackage{url}
\usepackage{amsfonts}
\usepackage{amssymb}
\usepackage{epsfig}
\usepackage{epstopdf}
\usepackage{slashbox}
\usepackage{bm}
\usepackage[margin=19mm,top=20mm,bottom=20mm]{geometry}

  {\proof}{\proofend}
\newtheorem{proposition}{Proposition}

\usepackage{fancyhdr}

\begin{document}

\begin{frontmatter}

\title{Performance Tradeoff in a Unified System of Communications and Passive Radar: A Secrecy Capacity Approach}
\author{Batu K. Chalise and Moeness G. Amin}
\tnotetext[]{B. K. Chalise is with the Department of Electrical and Computer Engineering, New York Institute of Technology, USA, email: batu.k.chalise@ieee.org.}
\tnotetext[]{M. G. Amin is with the Center for Advanced Communications, Villanova University, PA, USA, email: moeness.amin@villanova.edu.} 
\tnotetext[]{The corresponding Author's email address is batu.k.chalise@ieee.org.}
	\tnotetext[]{Please cite this article in press as: B. K. Chalise, M. G. Amin, ``Performance tradeoff in a unified system of communications and passive radar: A secrecy capacity approach,'' Digit. Signal Process. (2018), \url{https://doi.org/10.1016/j.dsp.2018.06.017}}
\begin{abstract}
In a unified system of  passive radar and communication systems of joint transmitter platform, information intended for a communication receiver may be eavesdropped by a passive radar receiver (RR), thereby undermining the security of communications system. To minimize this information security risk, in this paper, we propose a unified passive radar and  communications system wherein the signal-to-interference and noise ratio (SINR) at the RR is maximized while ensuring that the information secrecy rate is above a certain threshold value. We consider both scenarios wherein transmissions of the radar waveform and information signals are scheduled with the disjoint (non-overlapping case) as well as with the same set of resources (overlapping case). In both cases, the underlying optimization problems are non-convex. In the former case, we propose alternating optimization (AO) techniques that employ semidefinite programming and computationally efficient semi-analytical approaches. In the latter case,  AO method based on semi-definite relaxation approach is proposed to solve the optimization problem. By changing the threshold value of the information secrecy rate, we then characterize the performance tradeoff between passive radar and communication systems with the boundaries of the SINR-secrecy capacity regions. The performance comparison of the proposed optimization methods demonstrate the importance of  the semi-analytical approach and the advantage of overlapping case over non-overlapping one.  
\end{abstract}


\end{frontmatter}

\section{Introduction}  
Radar sensing and wireless communications are the two most prominent techniques that are based on similar radio frequency phenomena and can be characterized with similar signal processing techniques \cite{Davis}. However, a  radar system's typical goal is to detect, localize, and track targets, whereas the goal of communication systems is to maximize information transfer and enhance its reliability. Due to different objectives, hardware configurations, power and bandwidth requirements,  and frequency bands of operations, these two systems have been independently considered and developed as two separate entities. However,  due to an ever increasing number of wireless devices and networks as well as demand for high speed multimedia data services,  it is important for the two systems to share common spectrum and enhance bandwidth utilization via improved spectrum congestion techniques. In this regard,  some frequency spectrum, e.g.,  2-4 GHz range,  has been allocated for both radar and communication systems, such as  Long Term Evolution (LTE) \cite{3GPPspec1}. When two systems share the same frequency band, techniques  such as opportunistic spectrum sharing \cite{Correia}, dual-function radar-communications (DFRC)  \cite{Blunt}, \cite{ HassanienAmin}, and cooperation between radar and communication systems \cite{Surender}, \cite{Athina} have been proposed to minimize the inter-system interference and enhance the performance of both systems.\\
\hspace*{0.5cm} On the other hand,  passive radar systems (PRS) have received significant research interests due to their low cost, covertness, and availability of a large number of illumination sources, such as cellular base stations and television stations \cite{Poullin},  \cite{Salah1}. To this end, the authors of \cite{Chal2}-\cite{Ssuvedi} have proposed several algorithms for detecting, localizing, and tracking targets in PRS.  In \cite{Cui}, the detector based on the generalized likelihood ratio test (GLRT) has been proposed for PRS consisting of a single transmitter and a single receiver, whereas the corresponding GLRT detectors  for multiple-input multiple-output (MIMO) PRS have been developed  in \cite{Hack}-\cite{HackPhd}. While these papers assume multi-frequency networks,  an extension to single-frequency multi-static PRS has been proposed in \cite{ChaliseHimedSP}. \\
\hspace*{0.5cm} Recent advancements in PRS (especially in the case of single-frequency multi-static scenario)  demonstrate that the estimation of the non-cooperative transmitters'  waveforms  is challenging and significantly affects the  performance of the PRS. In particular, 
 the performance of the PRS approaches that of active radars \cite{ChaliseHimedSP}, if the waveform estimation is sufficiently accurate.  Motivated from this fact, the authors in \cite{SPL2017} propose to develop PRS as a part of a bandwidth-flexible communication system \cite{SPL2017},  where the transmitters no longer remain completely non-cooperative, and  in fact, assist the radar receiver in estimating the broadcast signals more efficiently through improved resource allocations. The single joint radar and communications transmitter proposed in \cite{SPL2017} is recently extended to a scenario of multiple transmitters in \cite{RadarConf2018}. However, in both papers,  information security is not considered, each transmitter is equipped with only a single antenna, and the radar and information signals are transmitted through orthogonal channels (non-overlapping case).

Security in wireless communications is a critical issue, since wireless channels are often prone to eavesdropping. To this end, based on the seminal work of  \cite{Wyner75},  information theoretic physical layer design approaches for enhancing security in wireless systems have been widely studied in the literature \cite{Hellman}-\cite{QiangLiJSAC}. Physical layer security approach aims to prevent unintended users from decoding information transmitted to the intended users by maximizing information {\it secrecy rate}. The advantage of this approach is that secrecy can be achieved  without using an encryption key. On the other hand, information theoretic metrics have been also used in the design and analysis of radar systems \cite{YangBlum}, \cite{Mathini}. To this end, the authors in \cite{Chambers2017} consider a  monostatic MIMO radar system,  wherein the objective is to enhance radar performance and secure information transmitted to a legitimate communication receiver from an eavesdropper-target. For this purpose, beamforming vectors, applied to communication and distortion signals, are jointly optimized. A Taylor series approximation approach \cite{Giannakis} is proposed to convexify the non-convex function of secrecy rate.  However, to the best of our knowledge, the problem of designing algorithms for a unified system of passive radar and communications, while emphasizing information security has not been investigated in the literature. This problem is important in a unified system since information signals intended for a communication receiver (CR) may be eavesdropped by a  passive radar receiver (RR), thereby undermining information security. Moreover, in contrast to \cite{Chambers2017},  our objective is to  jointly optimize radar waveforms and covariance matrix of information signals without additionally transmitting distortion signal.  

\hspace*{0.5cm}   In this paper, we consider a unified system consisting of  a transmitter, a passive RR, and a CR, each equipped with multiple antennas. The performance tradeoff  between radar and communications is characterized by obtaining the boundaries of the signal-to-interference-and-noise ratio (SINR) for the RR versus information secrecy rate region, when considering the same RR as an eavesdropper{\footnote{In general, information security should be achieved against all eavesdroppers, including the RR. While such design approach will be reported in our future work, it is worthwhile to mention that the RR's eavesdropping capability is higher than that of any other eavesdropper, since, as a part of the unified system, the RR has more knowledge about the settings, parameters, and protocols of the unified system.}}. To this end, joint optimization of radar waveforms and transmit covariance matrix of information signals is proposed with the objective of maximizing the SINR at the RR, while ensuring that the information secrecy rate is above a certain threshold. We formulate the underlying non-convex optimization problems and  provide corresponding solutions when the radar and information signals use both  orthogonal and non-orthogonal (overlapping) sets of resources. In both cases, iterative alternating optimization (AO) methods that employ semi-definite programming (SDP)/semi-definite relaxation (SDR) are proposed for optimizing radar waveforms and transmit covariance matrices{\footnote{These optimization techniques form the basis for solving several problems in other contexts, primarily in the design of communication only systems (see \cite{QiangLiJSAC}, \cite{ChaliseVandendorpe}, and references therein). We propose to  leverage these techniques for the joint transmitter design in a unified system of passive radar and communications.}}. However, in the former case, a computationally efficient semi-analytical approach is also proposed. Simulation results show that this approach provides significant performance gains over the SDP-based approach. Moreover,  in spite of interference caused in the overlapping method, due to joint optimization of radar waveforms and transmit covariance matrix, results show that the overlapping method provides better performance than the non-overlapping one.  

The remainder of this paper is organized as follows: Section \ref{sec1} presents the system model of unified passive radar and communications. Section \ref{sec2} provides problem formulations  and corresponding solutions for the optimization problems of the non-overlapping case.  The problem formulation and optimization method for the overlapping case are presented in Section \ref{sec3}. Numerical results are provided in Section \ref{sec4}. Finally, Section \ref{sec5} concludes the paper and summarizes the key findings.
 
{\it Notations:}  Upper (lower) bold face letters will be used for matrices (vectors);  $(\cdot)^{H}$, ${\bf I}_N$, $||{\cdot}||$, and  $\otimes$ denote Hermitian transpose,  $N\times N$  identity matrix, Euclidean norm for vector/Frobenius norm for matrix, and  Kronecker product operator, respectively. ${\rm tr}({\bf X})$ and  ${\rm det}({\bf X})$ denote trace and determinant of a matrix ${\bf X}$, respectively, ${\bf X}\succeq 0$ denotes that ${\bf X}$ is a positive semi-definite matrix, and ${\rm vec}({\bf X})$ denote the vectorization of ${\bf X}$. ${\mathcal C}^{N\times M}$ stands for a space of complex matrix of dimension $N\times M$, and ${\mathcal N_C}(\mu, \sigma^2)$ and ${\mathbb E}\{ \}$ denote circularly symmetric complex Gaussian distribution with mean $\mu$ and variance $\sigma^2$ and expectation operation, respectively. 

  
 \section{System Model}
 \label{sec1}
 Consider a system that supports both communications and radar receivers, as shown in Fig. \ref{fig:blkdiag1}. The transmitter and CR are equipped with $N_t$ and $M$ antennas, respectively.  The antennas of the  RR are divided into groups of direct channel (DC) antennas and surveillance channel (SC) antennas. Without loss of generality, we assume that the same, i.e., $N$ antennas are used for the DC and SC. The DC antennas receive signals via direct path from the transmitter, whereas the SC antennas receive signals originating from the transmitter but reflected by a target.  The direct path signal is used at the RR for estimating the radar waveform as in the case of PRS. Although the transmitter is not non-cooperative to the RR (in contrast to the conventional PRS), signal transmissions to the radar and communication receivers may be scheduled using totally disjoint, partially overlapping, and completely overlapping groups of resource elements (time-frequency units) \cite{3GPPspec1}. The optimum scheduling of the resource elements for the transmission of communication signal will change due to channel fading. This will, in turn, change the scheduling of  the resource elements for the transmission of radar waveforms, since the total available resource is the same in unified system. Moreover, since it becomes costly for the transmitter to  let the RR estimate the  transmitted radar waveform after every change of scheduling, we consider that  the RR estimates the radar waveform using  the direct path channel. We consider that the RR can achieve synchronization like the communication receivers  which achieve synchronization by detecting dedicated primary and secondary synchronization signals \cite{Sriharsha} transmitted by the transmitter, for example in LTE systems. As such, no additional communication links between the radar and communication receivers are required for maintaining synchronization. We also assume a clutter-free noise-only environment considering that the effect of the clutter-path signals can be mitigated by applying a variety of techniques (see \cite{Hack}-\cite{HackPhd} and references therein){\footnote{ The tradeoff analysis of the unified system  in the presence of clutters will be presented in our future work.}}.
 \begin{figure}[htb!] 
\centering
 \includegraphics[scale=0.5]{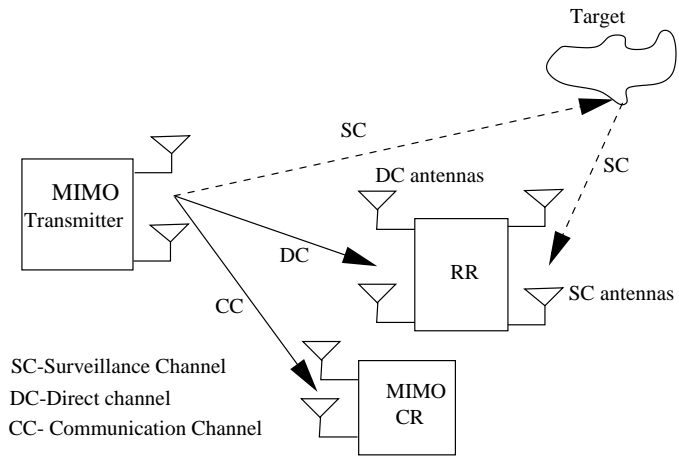}
\caption{A unified system with a transmitter, a RR and a CR, all equipped with multiple antennas } 
\label{fig:blkdiag1}
\end{figure}

The transmitter uses a portion of its total system power, $P_{T}$,  to broadcast the radar waveform  and the remaining portion to transmit  an information signal. We consider both orthogonal and non-orthogonal cases of signal transmissions. In the former, we assume that the signal transmissions from the transmitter are scheduled optimally using non-overlapping groups of resource elements. In the latter case, we relax this approximation and consider that the transmitter broadcasts radar and communication signals using the same resource elements. The tradeoff analysis is conducted by solving an optimization problem,  where the objective is to maximize the SINR at the RR while ensuring that the  information rate for the CR is above a certain threshold value. 
 
 The signal received by the RR, at a time instant $l$,  via direct path channel is expressed as
 \begin{eqnarray}
 \label{eqn1}
{\bf x}_d[l] =\gamma_d {\bf a}_d(\theta_r) {\bf a}_d^H(\theta_t){\bf s}_r[l]+{\bf v}_{dr}[l],
 \end{eqnarray}
 where $\theta_r$ is the direction of arrival (DoA) of the signal received via the direct path channel when reference is DC-antennas, and $\theta_t$ is the direction of departure (DoD) of the signal when reference is transmitter-antennas.  ${\bf a}_d(\theta_r) \in  {\mathcal C}^{N \times 1}$  and  ${\bf a}_d(\theta_t) \in {\mathcal C}^{N_t \times 1}$ denote the steering vectors corresponding to $\theta_r$ and $\theta_t$, respectively. $\gamma_d$ is the channel coefficient of the direct path channel between the transmitter and RR, and ${\bf s}_r[l] \in  {\mathcal C}^{N_t \times 1}$ is the radar waveform transmitted by the transmitter for the $l$th time instant. ${\bf v}_{dr}[l]  \in  {\mathcal C}^{N \times 1}$ denotes additive Gaussian noise at the DC-antennas of the RR, which is assumed to have zero mean and covariance of $\sigma_r^2 {\bf I}_{N}$, i.e., ${\mathcal N_C}({\bf 0}, \sigma_r^2 {\bf I}_{N})$. 
 
Collecting $l=1, \cdots, L$ vectors of ${\bf x}_d[l]$, we obtain matrix ${\bf X}_d=[{\bf x}_d[1], \cdots, {\bf x}_d[L]]$ of size $N\times L$. Let  ${\bf S}_r=[{\bf s}_r[1], \cdots, {\bf s}_d[L]] \in {\mathcal C}^{N_t\times L}$ and   ${\bf V}_{dr}=[{\bf v}_{dr}[1], \cdots, {\bf v}_{dr}[L]] \in {\mathcal C}^{N\times L}$. Then, ${\bf X}_d$ can be expressed as
\begin{eqnarray}
 \label{eqn2}
 {\bf X}_d={\bf H}_d{\bf S}_r+{\bf V}_{dr} \longrightarrow {\bf x}_d={\bf A}(\theta_d){\bf s}_r+{\bf v}_{dr}, 
\end{eqnarray}
where ${\bf H}_d=\gamma_d {\bf a}_d(\theta_r) {\bf a}_d^H(\theta_t)$, ${\bf x}_d={\rm vec}( {\bf X}_d) \in {\mathcal C}^{LN \times 1}$,  ${\bf s}_r={\rm vec}( {\bf S}_r) \in {\mathcal C}^{L N_t \times 1}$, ${\bf v}_{dr}={\rm vec}( {\bf V}_{dr}) \in {\mathcal C}^{LN \times 1}$,  ${\bf A}(\theta_d)=\gamma_d\left({\bf I}_L\otimes \left[ {\bf a}_d(\theta_r) {\bf a}_d^H(\theta_t) \right]\right) \in {\mathcal C}^{LN\times L N_t}$, and use the fact that ${\rm vec}({\bf H}_d{\bf S}_r{\bf I}_L)=({\bf I}_L \otimes {\bf H}_d){\bf s}_r$ \cite{MatrixCookbook}. 

On the other hand, the received signal vector at the $l$th time through surveillance channel can be expressed as
 \begin{eqnarray}
 \label{eqn3}
{\bf x}_s[l]=\gamma_{t} {\bf a}_{s}(\theta_{r,0}){\bf a}^H_{s}(\theta_{t,0}){\bf s}_r[l]+{\bf v}_{sr}[l],
\end{eqnarray}
where $\theta_{r,0}$  is the DoA of the signal received via the surveillance channel when reference is SC-antennas, and $\theta_{t,0}$ is the DoD of the signal when reference is transmitter-antennas.  ${\bf a}_s(\theta_{r,0}) \in  {\mathcal C}^{N \times 1}$  and  ${\bf a}_s(\theta_{t,0}) \in {\mathcal C}^{N_t \times 1}$ denote steering vectors corresponding to $\theta_{r,0}$ and $\theta_{t,0}$, respectively. $\gamma_t$ is the channel coefficient of the surveillance channel between the transmitter and RR and includes the effects of bi-static attenuation from the transmitter to the RR as well as target's reflection coefficient. ${\bf v}_{sr}[l]  \in  {\mathcal C}^{N \times 1}$ denotes additive Gaussian noise at the SC-antennas of the RR, which is also assumed to have zero mean and covariance of $\sigma_r^2 {\bf I}_{N}$, i.e., ${\mathcal N_C}({\bf 0}, \sigma_r^2 {\bf I}_{N})$. Following similar steps as in the derivation of (\ref{eqn2}) from  (\ref{eqn1}) and collecting $l=1, \cdots, L$ vectors of ${\bf x}_s[l]$ in a long vector ${\bf x}_s$, we have
 \begin{eqnarray}
 \label{eqn4}
{\bf x}_s={\bf A}_s(\theta_0){\bf s}_r+{\bf v}_{sr},
\end{eqnarray}
where  ${\bf A}_s(\theta_0)=\gamma_t\left( {\bf I}_L\otimes \left[ {\bf a}_{s}(\theta_{r,0}) {\bf a}_{s}^H(\theta_{t,0}) \right] \right)$, ${\bf v}_{sr}={\rm vec}\left[ {\bf v}_{sr}[1], \cdots, {\bf v}_{sr}[L] \right] \in {\mathcal C}^{LN\times 1}$, and 
${\bf x}_s={\rm vec}\left[ {\bf x}_{s}[1], \cdots, {\bf x}_{s}[L] \right] \in {\mathcal C}^{LN\times 1}$. The RR utilizes the direct path channel to estimate ${\bf s}_r$ by using matched filtering. As such, the estimated ${\hat {\bf s}}_r={\bf A}^H(\theta_d) {\bf x}_d$ can be expressed as
\begin{eqnarray}
 \label{eqn5}
{\hat {\bf s}}_r={\bf A}^H(\theta_d){\bf A}(\theta_d){\bf s}_r+{\bf A}^H(\theta_d){\bf v}_{dr}.
\end{eqnarray}
Noting that $({\bf A}\otimes {\bf B}) ({\bf C}\otimes {\bf D})=({\bf A C}\otimes {\bf B D})$ \cite{MatrixCookbook}, ${\bf A}^H(\theta_d){\bf A}(\theta_d)$ can be expressed as
\begin{eqnarray}
 \label{eqn6}
{\bf A}^H(\theta_d){\bf A}(\theta_d)=N|\gamma_d|^2\left( {\bf I}_L \otimes ({\bf a}_d(\theta_t) {\bf a}^H_d(\theta_t))\right)\triangleq {\bf A}_d(\theta_t).
\end{eqnarray}
Note that the estimated ${\hat {\bf s}}_r$ may not be noise free. The effect of noisy direct channel into radar's performance can be taken  into account by considering that this estimated ${\hat {\bf s}}_r$ is the true ${\bf s}_r$. Thus, substituting  (\ref{eqn5}) into   (\ref{eqn4}),  the received surveillance channel signal can be expressed as 
\begin{eqnarray}
\label{eqn7}
{\bf x}_s={\bf A}_s(\theta_0){\bf A}_d(\theta_t){\bf s}_r+  {\bf A}_s(\theta_0){\bf A}^H(\theta_d){\bf v}_{dr}+ {\bf v}_{sr}.
\end{eqnarray}
The signal ${\bf x}_s$ can be processed with a linear operator  ${\bf w} \in {\mathcal C}^{LN\times 1}$ which can be considered as a vectorized form of a spatio-temporal matrix of size $N\times L$. 
The resulting decision metric ${\tilde x}_s={\bf w}^H{\bf x}_s$  can be expressed as
\begin{eqnarray}
\label{eqn8}
{\tilde x}_s={\bf w}^H{\bf A}_s(\theta_0){\bf A}_d(\theta_t){\bf s}_r+  {\bf w}^H{\bf A}_s(\theta_0){\bf A}^H(\theta_d){\bf v}_{dr}+{\bf w}^H {\bf v}_{sr}.
\end{eqnarray}
As such, the SINR at the RR  is expressed as
 \begin{eqnarray}
\label{eqn9}
\gamma_{R}=\frac{1}{\sigma_r^2} \frac{ |{\bf w}^H{\bf A}_s(\theta_0){\bf A}_d(\theta_t){\bf s}_r|^2}{  ||{\bf w}^H{\bf A}_s(\theta_0){\bf A}^H(\theta_d)||^2+{\bf w}^H{\bf w}}.
\end{eqnarray}

 We assume block fading channel between the transmitter and CR, i.e., the  communication channel remains constant for a block of symbols and changes independently from one block to another. For conciseness, without loss of generality, this block length is considered to be $L$. Since we first consider the case in which signal transmissions to the RR and CR employ different set of non-overlapping resource units and the communication channel does not change over $L$ symbols, the  information rate for the CR can be obtained by considering the channel at a specific time instant{\footnote{As it will be clear later in the paper, that this is not true when radar and informations signals occupy the same set of resources.}}. As such, the signal received by the CR can be expressed as
\begin{eqnarray}
\label{eqn10}
{\bf x}_c={\bf H}_c {\bar {\bf s}}_c+{\bf v}_c,
\end{eqnarray}
where ${\bf H}_c \in {\mathcal C}^{M\times N_t}$ is the MIMO channel between the transmitter and CR,  ${\bar {\bf s}}_c \in {\mathcal C}^{N_t\times 1}$ is the vector of data transmitted from $N_t$ antennas at a given time instant, and ${\bf v}_c \in {\mathcal C}^{M \times 1}$ is additive Gaussian noise at the antennas of CR. Each element of  ${\bf v}_c $ is assumed to be distributed as ${\mathcal N_C}(0, \sigma_c^2)$. Since information intended for the CR will be also received by the RR, any confidential information can be decoded by the RR, thereby undermining the security of the communication system. Considering that the direct channel is much stronger than the surveillance channel, the information received by the RR through the direct path channel can undermine the security most. Therefore, we mainly focus on the information signal leaked by the transmitter to the RR through the direct path channel.  The received information signal at the RR  is given by
\begin{eqnarray}
\label{eqn11}
{\tilde {\bf x}}_r={\bf H}_d {\bar {\bf s}}_c+{\bar {\bf v}}_{dr},
\end{eqnarray}
 ${\bar {\bf v}}_{dr} \in {\mathcal C}^{N \times 1}$ is the additive Gaussian noise at the DC-antennas of the RR. The capacity of the transmitter-CR channel is given by
\begin{eqnarray}
\label{eqn12}
C_c=\log_2\left( {\rm det}\left(\sigma_c^{-2} {\bf H}_c{\bf Q}_c {\bf H}_c^H+{\bf I}_{M}\right)\right),
\end{eqnarray}
where ${\bf Q}_c={\mathbb E}\left\{ {\bar {\bf s}}_c{\bar {\bf s}}_c^H\right\} \in {\mathcal C}^{N_t \times N_t}$. The capacity of the transmitter-RR channel is given by
\begin{eqnarray}
\label{eqn13}
C_r=\log_2\left( {\rm det}\left(\sigma_r^{-2} {\bf H}_d{\bf Q}_c {\bf H}_d^H+{\bf I}_{N}\right)\right). 
\end{eqnarray}
The secrecy capacity is given by $C_s=\max(0, C_c-C_r)$ \cite{Hassibi} , where
\begin{eqnarray}
\label{eqn13N}
 C_s=\max\left(0,\log_2\left( {\rm det}\left(\sigma_c^{-2}{\bf H}_c{\bf Q}_c {\bf H}_c^H+{\bf I}_{M}\right)\right)- \log_2\left( {\rm det}\left( \sigma_r^{-2} {\bf H}_d{\bf Q}_c {\bf H}_d^H+{\bf I}_{N}\right)\right)\right).
 \end{eqnarray}
 
\section{Proposed Optimization}
\label{sec2}
Given a total system power, $P_T$, the objective is to maximize the received SINR at the RR while maintaining the secrecy capacity above a certain threshold value. We assume that the channels from the transmitter to the CR and RR (direct path) can be estimated with sufficient accuracy using channel acquisition techniques proposed in \cite{ChaliseZhangAminTSP}. Moreover, we assume that noise powers at all receiver terminals are known. Since these parameters may not be perfectly known, the performance results presented in this paper will serve as upper bounds for the performance of the underlying system in practice. The optimization will be solved for each hypothesized target position (or equivalently range-Doppler cell) and updated after  each coherence time of the communication channel. The secrecy rate threshold is a user specific parameter and depends on the requested level of security, i.e., usually a  larger value of the threshold is selected for a communication link requiring the higher priority in information security.
Mathematically, the proposed optimization can be expressed as
\begin{eqnarray}
\label{eqn14}
\max_{{\bf w}, {\bf s}_r, {\bf Q}_c \succeq 0} & & \gamma_R \nonumber\\
{\rm s.t.} & & C_s\geq r_m, \\
& & {\rm tr}({\bf Q}_c)+{\bf s}_r^H{\bf s}_r \leq P_T, \nonumber
\end{eqnarray} 
where $r_m$ is the threshold value of information secrecy rate. The constraints of the optimization problem (\ref{eqn14})  do not depend on ${\bf w}$. For a given ${\bf s}_r$, the optimization w.r.t. ${\bf w}$ can be expressed as
\begin{eqnarray}
\label{eqn15}
\max_{{\bf w} } \frac{{\bf w}^H {\bf A}_s(\theta_0){\bf A}_d(\theta_t){\bf s}_r{\bf s}_r^H {\bf A}^H_d(\theta_t) {\bf A}^H_s(\theta_0) {\bf w}}{{\bf w}^H \left[{\bf A}_s(\theta_0){\bf A}_d(\theta_t){\bf A}^H_s(\theta_0)+{\bf I}_{LN}\right] {\bf w}}=
\max_{{\bf w} } \frac{|{\bf s}_r^H {\bf D}{\bf w}|^2}{{\bf w}^H {\bf C}{\bf w}}
\end{eqnarray}
where
\begin{eqnarray}
\label{eqn16}
{\bf C}&\triangleq & {\bf A}_s(\theta_0){\bf A}_d(\theta_t){\bf A}_s^H(\theta_0)+{\bf I}_{LN}, \nonumber\\
{\bf D} &\triangleq & {\bf A}_d^H(\theta_t) {\bf A}_s^H(\theta_0).
\end{eqnarray}
This maximization problem in (\ref{eqn15}) can be equivalently expressed as
\begin{eqnarray}
\label{eqn17}
\min_{\bf w} & &  {\bf w}^H{\bf C}{\bf w} \nonumber\\
{\rm s.t.} & & {\bf s}_r^H {\bf D} {\bf w}=1. 
\end{eqnarray}
Employing Lagrangian multiplier function approach \cite{MohammadChalise}, it can be shown that the optimum ${\bf w}$ in (\ref{eqn17}) is
\begin{eqnarray}
\label{eqn18}
{\bf w}=\frac{{\bf C}^{-1}{\bf D}^H{\bf s}_r}{{\bf s}_r^H{\bf D}{\bf C}^{-1} {\bf D}^H{\bf s}_r}. 
\end{eqnarray}
Substituting (\ref{eqn18}) into $\gamma_R$,  the resulting SINR can be expressed as
\begin{eqnarray}
\label{eqn19}
\gamma_R=\frac{1}{\sigma_r^2} {\bf s}_r^H{\bf D}{\bf C}^{-1}{\bf D}^H{\bf s}_r. 
\end{eqnarray}
Without loss of generality, we assume that $r_m\geq 0$. This means that the constraint $C_s\geq r_m$ can be cast as $C_c-C_r \geq 0$. Substituting (\ref{eqn19}) into (\ref{eqn14}), we obtain the following optimization
\begin{eqnarray}
\label{eqn20}
\max_{{\bf s}_r, {\bf Q}_c \succeq 0} & & {\bf s}_r^H{\bf D}{\bf C}^{-1}{\bf D}^H{\bf s}_r \nonumber\\
{\rm s.t.} & &  \biggl\{ \log_2\left( {\rm det}\left({\bf H}_c{\bf Q}_c {\bf H}_c^H+\sigma_c^2{\bf I}_{M}\right)\right)-\log_2\left( {\rm det}\left({\bf H}_d{\bf Q}_c {\bf H}_d^H+\sigma_r^2{\bf I}_{N}\right)\right) + c_a \biggr\} \geq r_m, \\
& & {\rm tr}({\bf Q}_c)+{\bf s}_r^H{\bf s}_r \leq P_T, \nonumber
\end{eqnarray}
where $c_a=\log_2\sigma_r^{2N}-\log_2\sigma_c^{2M}$ and the first constraint is not a function of  temporal structure of ${\bf s}_r$. Let ${\bf s}_r=\sqrt{P_r}{\bar {\bf s}}_r$, where ${\bar {\bf s}}^H_r{\bar {\bf s}}_r=1$ and $P_r$ is the power allocated for the radar waveform. Then, it is clear that the optimum ${\bar {\bf s}}_r$ is the eigenvector corresponding to the largest eigenvalue of ${\bf D}{\bf C}^{-1}{\bf D}^H$.  In this case, the objective function of (\ref{eqn20}) can be expressed as 
\begin{eqnarray}
\label{eqn21}
{\bf s}_r^H{\bf D}{\bf C}^{-1}{\bf D}^H{\bf s}_r =P_r\lambda_{\max}\left( {\bf D}{\bf C}^{-1}{\bf D}^H \right),
\end{eqnarray}
where $\lambda_{\max}(\cdot)$ stands for the maximum eigenvalue of a matrix. Note that the SINR at the RR does not depend on actual ${\bf Q}_c$, but only on its trace since $P_r$ is a function of ${\rm tr}({\bf Q}_c)$. Substituting (\ref{eqn21}) into (\ref{eqn20}),  it can be expressed in terms of $P_r$ and ${\bf Q}_c$ as 
\begin{eqnarray}
\label{eqn22}
\max_{P_r, {\bf Q}_c \succeq 0} & & P_r \nonumber\\
{\rm s.t.} & &  \biggl\{ \log_2\left( {\rm det}\left({\bf H}_c{\bf Q}_c {\bf H}_c^H+\sigma_c^2{\bf I}_{M}\right)\right)-\log_2\left( {\rm det}\left({\bf H}_d{\bf Q}_c {\bf H}_d^H+\sigma_r^2{\bf I}_{N}\right)\right)+c_a\biggr\} \geq r_m, \\
& & {\rm tr}({\bf Q}_c)+P_r \leq P_T. \nonumber
\end{eqnarray}
It is clear that the optimum $P_r$ is such that $P_r=P_T- {\rm tr}({\bf Q}_c)$. Substituting $P_r$ in (\ref{eqn22}), the remaining optimization in terms of ${\bf Q}_c$ can be expressed as
\begin{eqnarray}
\label{eqn23}
\min_{{\bf Q}_c \succeq 0} & & {\rm tr}({\bf Q}_c) \nonumber\\
{\rm s.t.} & &  \biggl\{ \log_2\left( {\rm det}\left({\bf H}_c{\bf Q}_c {\bf H}_c^H+\sigma_c^2{\bf I}_{M}\right)\right)- \log_2\left( {\rm det}\left({\bf H}_d{\bf Q}_c {\bf H}_d^H+\sigma_r^2{\bf I}_{N}\right)\right)+c_a\biggr\} \geq r_m.
\end{eqnarray}
Note that if the optimum ${\bf Q}_c$ is such that ${\rm tr}({\bf Q}_c)$ turns out to be larger than $P_T$, then such solution is not feasible. Moreover, the secrecy rate does not depend on actual ${\bf s}_r$ and only on its squared norm (power). 
Unfortunately, the optimization problem (\ref{eqn23}) is not convex. Following the approach of maximizing the secrecy capacity under a transmit power constraint in a MIMO wiretap channel \cite{QiangLiJSAC}, we propose alternating optimization (AO) methods for solving  (\ref{eqn23}). We first propose an iterative approach wherein an SDP problem is solved in each iteration. We then propose an AO method where the SDP optimization will be replaced by a semi-analytical approach that includes bisection method \cite{QiangLiJSAC}, \cite{ChaliseVandendorpe}.  

\subsection{Iterative SDP}
Since ${\rm det}({\bf A}^{-1})=\frac{1}{{\rm det}({\bf A})}$ \cite{MatrixCookbook}, (\ref{eqn23}) can be expressed as
\begin{eqnarray}
\label{eqn24}
\min_{{\bf Q}_c \succeq 0} \hspace*{-0.6cm} & & {\rm tr}({\bf Q}_c) \nonumber\\
{\rm s.t.} \hspace*{-0.6cm} & &  \biggl\{ \log_2\left( {\rm det}\left({\bf H}_c{\bf Q}_c {\bf H}_c^H+\sigma_c^2{\bf I}_{M}\right)\right)+
  \log_2\left( {\rm det}\left({\bf H}_d{\bf Q}_c {\bf H}_d^H+\sigma_r^2{\bf I}_{N}\right)^{-1}\right) +c_a\biggr\} \geq r_m.
\end{eqnarray}
We now introduce an additional matrix variable ${\bf Y}\in {\mathcal C}^{N \times N}, {\bf Y}\succeq 0$, and utilize the following expression \cite{Rangarajan}:
\begin{eqnarray}
\label{eqn25}
\hspace*{-0.8cm} & &  \hspace*{-0.8cm} \log\left( {\rm det}\left({\bf H}_d{\bf Q}_c {\bf H}_d^H+\sigma_r^2{\bf I}_{N}\right)^{-1}\right) =\max_{{\bf Y}\succeq 0} \biggl\{\log({\rm det}({\bf Y}))- 
 {\rm tr}\left( {\bf Y} \left({\bf H}_d{\bf Q}_c {\bf H}_d^H+\sigma_r^2{\bf I}_{N}\right)\right) + N\biggr\}.
\end{eqnarray}
Substituting (\ref{eqn25}) into (\ref{eqn24}), we obtain the following optimization problem:
\begin{eqnarray}
\label{eqn26}
\min_{{\bf Q}_c \succeq 0} & & \hspace*{-0.4cm} {\rm tr}({\bf Q}_c) \nonumber\\
{\rm s.t.} & & \hspace*{-0.4cm} \max_{{\bf Y}\succeq 0} \biggl\{  \log\left( {\rm det}\left({\bf H}_c{\bf Q}_c {\bf H}_c^H+\sigma_c^2{\bf I}_{M}\right)\right)- \biggr.\nonumber\\
\biggl. & & 
{\rm tr}\left( {\bf Y} \left({\bf H}_d{\bf Q}_c {\bf H}_d^H+\sigma_r^2{\bf I}_{N}\right)\right)+\log({\rm det}({\bf Y}))+{\bar N}\biggr\} \geq {\bar r}_m,
\end{eqnarray}
where  ${\bar N}=c_a+N$ and ${\bar r}_m=r_m \log(2)$. The optimization  problem (\ref{eqn26}) can be solved in an iterative way as follows. For a given ${\bf Q}_c$, the optimum ${\bf Y}$ can be obtained by solving the first-order derivative of the constraint w.r.t. to ${\bf Y}$. This leads to
 \begin{eqnarray}
 \label{eqn27}
 {\bf Y}= \left({\bf H}_d{\bf Q}_c {\bf H}_d^H+\sigma_r^2{\bf I}_N\right)^{-1}. 
 \end{eqnarray}
 On the other hand, for a given ${\bf Y}\succeq 0$, the optimization over ${\bf Q}_c\succeq 0$ can be expressed as
  \begin{eqnarray}
 \label{eqn28}
\min_{{\bf Q}_c \succeq 0} & & \hspace*{-0.4cm} {\rm tr}({\bf Q}_c) \nonumber\\
{\rm s.t.} & & \hspace*{-0.4cm} \biggl\{  \log\left( {\rm det}\left({\bf H}_c{\bf Q}_c {\bf H}_c^H+\sigma_c^2{\bf I}_{M}\right)\right)- 
{\rm tr}\left( {\bf Y} \left({\bf H}_d{\bf Q}_c {\bf H}_d^H+\sigma_r^2{\bf I}_{N}\right)\right)+\biggr. \nonumber\\
\biggl. & &  \log({\rm det}({\bf Y}))+{\bar N}\biggr\} \geq {\bar r}_m,
 \end{eqnarray}
This optimization problem is convex and can be solved numerically using convex optimization toolbox such as discipline convex programming (CVX) \cite{GrantCVX}. The algorithm ({\bf Algorithm 1}) to solve the optimization problem (\ref{eqn23}) is summarized below.
\begin{itemize}
\item 1) Initialize maximum number of iterations, convergence accuracy, $\epsilon$, and initial ${\bf Q}_c\succeq 0$.
\item 2) Update ${\bf Y}$ using  (\ref{eqn27}).
\item 3) Update ${\bf Q}_c$ by solving  (\ref{eqn28}).
\item 4) Go to step (2) until required convergence accuracy is achieved or maximum number of iterations is reached. 
\end{itemize}
Note that the SDP problem (\ref{eqn28}) (which is convex optimization problem over ${\bf Q}_c$ )  is solved for a given ${\bf Y}$.  In each iteration of the algorithm, ${\bf Y}$ is updated to approximate the left-hand side of (\ref{eqn25}) via maximization of a concave function over ${\bf Y}$. This means that in  each iterative step of {\bf Algorithm 1}, we have an improved estimate of the left-hand side of the constraint used in the original problem (\ref{eqn23}). This leads to a decreasing objective function in each iteration.
 Moreover, since the objective function is continuously differentiable, and each variable (i.e, ${\bf Q}_c$ and ${\bf Y}$) belongs to a nonempty, closed, and concave subset,  the AO approach is guaranteed to converge \cite{QiangLiJSAC}. However, in general the execution of the SDP problems become very slow when they consist of large size matrices. This is evident from the worst-case complexity of the standard form SDPs, which is  given by ${\mathcal O}\left(N_t^{4.5}\log\left( \frac{1}{\epsilon}\right)\right)$ \cite{LuoMaTSPM10} for a given  solution accuracy of $\epsilon$. As such, the worst-case complexity of   {\bf Algorithm 1} is  larger than ${\mathcal O}\left(N_{it} N_t^{4.5}\log\left( \frac{1}{\epsilon}\right)\right)$, where $N_{it}$ is the number of iterations required to achieve $|{\rm tr}({\bf Q}_c^{(n)})- {\rm tr}({\bf Q}_c^{(n-1)})|\leq \epsilon$, where ${\bf Q}_c^{(n)}$ denotes the covariance matrix at the $n$th iteration. 

\subsection{Semi-analytical Approach}
Motivated from the solution approach of the secrecy rate maximization problem in MIMO system \cite{QiangLiJSAC}, we propose a semi-analytical approach for solving (\ref{eqn28}).  The Lagrangian multiplier function for  (\ref{eqn28}) can be expressed as
  \begin{eqnarray}
 \label{eqn29}
& & {\mathcal L}({\bf Q}_c, \lambda)={\rm tr}({\bf Q}_c)+\lambda\biggl\{   {\bar r}_m- \log\left( {\rm det}({\bf Y})\right)-{\bar N}+ 
{\rm tr}\left( {\bf Y} \left({\bf H}_d{\bf Q}_c {\bf H}_d^H+\sigma_r^2{\bf I}_{N}\right)\right)-\biggr. \nonumber\\
\biggl. & & 
\hspace*{2cm}  \log\left( {\rm det}\left({\bf H}_c{\bf Q}_c {\bf H}_c^H+\sigma_c^2{\bf I}_{M}\right)\right)\biggr\}, 
 \end{eqnarray}
 where $\lambda\geq 0$ is a Lagrangian multiplier. Now the main result is presented in the following proposition. 
\begin{proposition}
For a given feasible $\lambda$, the optimum solution of ${\bf Q}_c$, as a function of $\lambda$,  is given by
\begin{eqnarray}
\label{eqn30}
{\bf Q}_c(\lambda)={\bf P}^{-\frac{H}{2}} {\bf V}{\boldsymbol \Lambda} {\bf V}^H {\bf P}^{-\frac{1}{2}},
\end{eqnarray}
where ${\bf P}={\bf I}_{N_t}+\lambda {\bf H}_d^H{\bf Y}{\bf H}_d$, ${\bf V}$ is a matrix of left singular vectors of ${\bf H}_c {\bf P}^{-\frac{H}{2}}$, i.e., ${\bf H}_c {\bf P}^{-\frac{H}{2}}={\bf U}{\boldsymbol \Sigma}{\bf V}^H$, 
and the non-zero diagonal elements, $\{\mu_i\}_{i=1}^{r}$ of ${\boldsymbol \Lambda}$, are given by
\begin{eqnarray}
\label{eqn40}
\mu_i=\left[ \lambda-\frac{\sigma_c^2}{d_i^2}\right]^{+},
\end{eqnarray}
where $\left\{d_i \right\}_{i=1}^{r}$ are the non-zero diagonal elements of ${\boldsymbol \Sigma}$, and $r=\min(M, N_t)$. 
\end{proposition}
\begin{proof}
Please refer to Appendix. 
\end{proof}
The remaining step is to calculate the value of $\lambda$. The optimum $\lambda$ is such that it satisfies the following complementary slackness condition \cite{Boyd}
  \begin{eqnarray}
 \label{eqn41}
& & \lambda\biggl\{   {\bar r}_m- \log\left( {\rm det}({\bf Y})\right)-{\bar N}+
{\rm tr}\left( {\bf Y} \left({\bf H}_d{\bf Q}_c(\lambda) {\bf H}_d^H+\sigma_r^2{\bf I}_{N}\right)\right)-\biggr. \nonumber\\
\biggl. & &  \log\left( {\rm det}\left({\bf H}_c{\bf Q}_c(\lambda) {\bf H}_c^H+\sigma_c^2{\bf I}_{M}\right)\right)\biggr\}=0.
 \end{eqnarray}
 The optimum $\lambda$ cannot be equal to zero. This is obvious since $\frac{\partial {\mathcal L}({\bf Q}_c, \lambda)}{\partial {\bf Q}_c}={\bf I}_{N_t}$, i.e., the partial derivative of the Lagrangian function with respect to ${\bf Q}_c$
  cannot be a zero matrix for $\lambda=0$. As such, the optimum $\lambda$ is such that 
  \begin{eqnarray}
 \label{eqn42}
& & g(\lambda)\triangleq \biggl\{   {\bar r}_m- \log\left( {\rm det}({\bf Y})\right)-{\bar N}+ 
{\rm tr}\left( {\bf Y} \left({\bf H}_d{\bf Q}_c(\lambda) {\bf H}_d^H+\sigma_r^2{\bf I}_{N}\right)\right)-\biggr. \nonumber\\
\biggl. & &  \log\left( {\rm det}\left({\bf H}_c{\bf Q}_c(\lambda) {\bf H}_c^H+\sigma_c^2{\bf I}_{M}\right)\right)\biggr\}=0.
 \end{eqnarray}
The value of $\lambda$ can be found by solving  (\ref{eqn42}). As there exists no closed-form solution for this equation, $\lambda$ can be obtained from a general one-dimensional search over $\lambda$ or more specifically, the bisection method. Thus, the proposed semi-analytical approach for solving  the optimization problem (\ref{eqn23}) is summarized below ({\bf Algorithm 2}).
\begin{itemize}
\item 1) Initialize maximum number of iterations, $N_{it}$, convergence accuracy, $\epsilon$, and ${\bf Q}_c\succeq 0$
\item 2) Update ${\bf Y}$ using  (\ref{eqn27})
\item 3) Execute the following steps of bisection method to find $\lambda$
\begin{itemize}
\item a)  Initialize $\lambda_{\min}$ and $\lambda_{\max}$ such that $g(\lambda_{\min}) g(\lambda_{\max})<0$
\item b) Set $\lambda_n=\frac{\lambda_{\min}+\lambda_{\max}}{2}$
\item c) Calculate ${\bf Q}_c(\lambda_n)$ by using  (\ref{eqn30})
\item d) If $g(\lambda_n)g(\lambda_{\max})< 0$, set $\lambda_{\min}=\lambda_n$,  otherwise set $\lambda_{\max}=\lambda_n$
\item e) Go to step (b) until convergence of bisection algorithm 
\end{itemize}
\item 4) Go to step (2) until required convergence accuracy is achieved or maximum number of iterations is reached
\end{itemize}
Note that the steps of bisection method to find $\lambda$  can be alternatively implemented by using {\it one-dimensional grid search} over $\lambda$. Using similar arguments as in the case of {\bf Algorithm 1}, the convergence of this algorithm can be guaranteed. For a given ${\bf Y}$, the bisection algorithm is guaranteed to converge, which requires $n_b=\log_2\left( \frac{  \lambda_{\max}^{(0)}- \lambda_{\min}^{(0)}}{\epsilon} \right)$ iterations \cite{ChaliseVandendorpe}, where  $\lambda_{\max}^{(0)}- \lambda_{\min}^{(0)}$ is the initial interval of $\lambda$ within which the root of $g(\lambda)$ lies. Therefore, the computational complexity of  {\bf Algorithm 2} is given by ${\mathcal O}\left(N_{it}n_b\right)$, where $N_{it}$ is the number of outer iterations (i.e., iterations over ${\bf Y}$).  Clearly, in contrast to {\bf Algorithm 1},  the complexity of  {\bf Algorithm 2} does not increase polynomially (with an exponent of 4.5) in $N_t$. This makes  possible the execution of {\bf Algorithm 2}  much faster than  {\bf Algorithm 1}. Moreover,  it is worthwhile to mention that this type of  AO methods guarantee suboptimum solutions and, in fact,  converge to Karush-Kuhn-Tucker (KKT) point. This can be proven by following similar derivations as in the case of MIMO wiretap channels (see Proposition 1, \cite{QiangLiJSAC} and the references therein). On the other hand, the global optimization technique, such as branch and bound \cite{Honggang},  can be applied to get the global optimum solution. However, such global optimization method requires exponential complexity and becomes computationally prohibitive in radar systems in which $L$ can easily take large values. On the other hand, a rigorous investigation, which we believe is a significant new task and beyond the scope of this paper, is required for solving the underlying optimization problems with the global optimization techniques \cite{Honggang}.

\section{Radar and Communication Transmissions with Same Resources}
 \label{sec3}
In this section, we consider that the transmitter sends radar and information signals using the same set of resources (overlapping case). In this case, both RR and CR observes a mixture of radar and information signals. Due to this reason, we will find in the sequel that the optimization problem required for obtaining the optimum tradeoff between radar and communication systems become further challenging to solve. In this overlapping case,  the transmitted signal vector for the $l$th time instant can be expressed as
\begin{eqnarray}
\label{eqn43}
{\bf s}[l]={\bf s}_c[l]+{\bf s}_r[l], l=1,\cdots, L,
\end{eqnarray}
${\bf s}_c[l] \in {\mathcal C}^{N_t \times 1} $  is the communication signal corresponding to the $l$th time instant. Replacing ${\bf s}_r[l]$ by ${\bf s}[l]$ in  (\ref{eqn1}) and following similar derivations as in  (\ref{eqn2})-(\ref{eqn6}), the received signal vector at the surveillance antennas, ${\bar {\bf x}}_s\in {\mathcal C}^ {LN\times 1}$  (i.e, equivalent version of (\ref{eqn7})) can be expressed as
\begin{eqnarray}
\label{eqn44}
{\bar {\bf x}}_s={\bf A}_s(\theta_0){\bf A}_d(\theta_t)({\bf s}_r+{\bf s}_c)+{\bf A}_s(\theta_0){\bf A}^H(\theta_d){\bf v}_{dr}+ {\bf v}_{sr},
\end{eqnarray}
where ${\bf s}_c={\rm vec}\left[ {\bf s}_c[1], \cdots, {\bf s}_c[L]\right] \in {\mathcal C}^{LN_t \times 1}$. After linear processing of the received surveillance signal with spatio-temporal vector ${\bar {\bf w}}\in{\mathcal C}^{LN\times 1}$ , the SINR at the RR is given by
\begin{eqnarray}
\label{eqn45}
\gamma_R=\frac{|{\bar {\bf w}}^H{\bf A}_s(\theta_0){\bf A}_d(\theta_t) {\bf s}_r|^2}{{\bar {\bf w}}^H{\bf C}({\bf Q}_c){\bar {\bf w}}},
\end{eqnarray}
where
\begin{eqnarray}
\label{eqn46}
{\bf C}({\bf Q}_c)={\bf A}_s(\theta_0){\bf A}_d(\theta_t)({\bf I}_L \otimes {\bf Q}_c){\bf A}^H_d(\theta_t){\bf A}^H_s(\theta_0)+\sigma_r^2{\bf I}_{LN}+\sigma_r^2{\bf A}_s(\theta_0){\bf A}_d(\theta_t) {\bf A}^H_s(\theta_0).
\end{eqnarray}
Considering that the channels remain same during a period of $L$ time instants, the signals received by the DC-antennas of the RR and CR are, respectively, given by
\begin{eqnarray}
\label{eqn47}
{\bar {\bf x}}_r&=&{\bar {\bf H}}_d {\bf s}_c+{\bar {\bf H}}_d{\bf s}_r+{\bf v}_{dr}, \nonumber\\
{\bar {\bf x}}_c&=& {\bar {\bf H}_c} {\bf s}_c+{\bar {\bf H}}_c{\bf s}_r+{\bar {\bf v}}_{c},
\end{eqnarray}
where ${\bar {\bf H}}_d={\bf I}_L\otimes {\bf H}_d$,  ${\bar {\bf H}}_c={\bf I}_L\otimes {\bf H}_c$, and ${\bar {\bf v}}_{c}\in {\mathcal C}^{L M \times 1}$ is zero-mean additive Gaussian noise with the variance $\sigma_c^2$. 
Define ${\bf R}_c \triangleq {\bar {\bf H}}_c{\bf s}_r{\bf s}_r^H{\bar {\bf H}}_c^H+\sigma_c^2{\bf I}_{LM}$, ${\bf R}_d \triangleq {\bar {\bf H}}_d{\bf s}_r{\bf s}_r^H{\bar {\bf H}}_d^H+\sigma_d^2{\bf I}_{LN}$. The capacities for the transmitter-CR and transmitter-RR, are, respectively given by
\begin{eqnarray}
\label{eqn48}
{\tilde C}_c&=&\log_2({\rm det}\left({\bf I}_{LM}+{\bar {\bf H}}_c({\bf I}_L \otimes {\bf Q}_c) {\bar {\bf H}}_c^H{\bf R}_c^{-1}\right)), \nonumber\\
{\tilde C}_d&=&\log_2({\rm det}\left({\bf I}_{LN}+{\bar {\bf H}}_d({\bf I}_L \otimes {\bf Q}_c) {\bar {\bf H}}_d^H{\bf R}_d^{-1}\right)).
\end{eqnarray}
The secrecy capacity can be expressed as ${\tilde C}_s=\max(0, {\tilde C}_c-{\tilde C}_d)$, where  ${\tilde C}_c-{\tilde C}_d \triangleq {\bar C}_s$ can be expressed as
\begin{eqnarray}
\label{eqn49}
{\bar C}_s&=&\log_2({\rm det}\left({\bf R}_c+{\bar {\bf H}}_c({\bf I}_L \otimes {\bf Q}_c) {\bar{\bf H}}_c^H\right))+\log_2({\rm det}({\bf R}_c^{-1}))+\nonumber\\
& & \log_2({\rm det}({\bf R}_d)) -\log_2({\rm det}\left({\bf R}_d+{\bar {\bf H}}_d({\bf I}_L \otimes {\bf Q}_c) {\bar {\bf H}}_d^H\right)).
\end{eqnarray}
As in the case where radar and communications signals occupy different resources, our objective is to optimize SINR at the RR, while ensuring that the secrecy capacity is above a certain threshold value, ${\tilde r}_m$.  
\begin{eqnarray}
\label{eqn50}
\max_{{\bar {\bf w}}, {\bf s}_r, {\bf Q}_c\succeq 0} & & {\bar \gamma}_R\triangleq \frac{{\bar {\bf w}}^H{\bf D}^H{\bf s}_r{\bf s}_r^H{\bf D}{\bar {\bf w}}}{{\bar{\bf w}}^H{\bf C}({\bf Q}_c){\bar{\bf w}}} \nonumber\\
{\rm s.t.} & & {\tilde C}_s \geq {\tilde r}_m,\\
&& {\bf s}_r^H{\bf s}_r+{\rm tr}({\bf Q}_c)\leq P_T. \nonumber
\end{eqnarray}
For a given ${\bf Q}_c\succeq {\bf 0}$ and ${\bf s}_r$, the objective function in (\ref{eqn50}) is a function of only ${\bar {\bf w}}$. The optimum ${\bar {\bf w}}$ will be similar as in  (\ref{eqn18}) with ${\bf C}$ replaced by ${\bf C}({\bf Q}_c)$. 
After substituting such optimum ${\bar {\bf w}}$, the objective function in (\ref{eqn50}) turns to
\begin{eqnarray}
\label{eqn51}
{\bar \gamma}_R={\bf s}_r^H{\bf D} {\bf C}({\bf Q}_c)^{-1}{\bf D}^H{\bf s}_r. 
\end{eqnarray}
As such, the optimization problem  (\ref{eqn50}) can be re-expressed as
\begin{eqnarray}
\label{eqn52}
\max_{{\bf s}_r, {\bf Q}_c\succeq 0} & &  {\bf s}_r^H{\bf D} {\bf C}({\bf Q}_c)^{-1}{\bf D}^H{\bf s}_r \nonumber\\
{\rm s.t.} & & {\tilde C}_s \geq {\tilde r}_m, \\
&& {\bf s}_r^H{\bf s}_r+{\rm tr}({\bf Q}_c)\leq P_T. \nonumber
\end{eqnarray}
Note that $\max_{\{ {\bf s}_r, {\bf Q}_c\succeq 0 \}}  {\bar \gamma}_R=\max_{\{{\bf s}_r, {\bf Q}_c\succeq 0\}} \left[1+ {\bf s}_r^H{\bf D} {\bf C}({\bf Q}_c)^{-1}{\bf D}^H{\bf s}_r\right]$. Since
\begin{eqnarray}
\label{eqn53}
1+ {\bf s}_r^H{\bf D} {\bf C}({\bf Q}_c)^{-1}{\bf D}^H{\bf s}_r={\rm det}({\bf I}_{LN}+{\bf D}^H{\bf s}_r{\bf s}_r^H{\bf D}{\bf C}({\bf Q}_c)^{-1}), 
\end{eqnarray}
and it can be easily shown that
\begin{eqnarray}
\label{eqn54}
 \log\left( {\rm det}({\bf I}_{LN}+{\bf D}^H{\bf s}_r{\bf s}_r^H{\bf D}{\bf C}({\bf Q}_c)^{-1})\right)=\log\left({\rm det}({\bf C}({\bf Q}_c)^{-1})\right) +\log\left({\rm det}({\bf C}({\bf Q}_c)+{\bf D}^H{\bf s}_r{\bf s}_r^H{\bf D})\right), 
\end{eqnarray}
the optimization problem (\ref{eqn52}) is expressed as
\begin{eqnarray}
\label{eqn55}
\max_{{\bf s}_r, {\bf Q}_c \succeq 0} & & \log\left({\rm det}({\bf C}({\bf Q}_c)+{\bf D}^H{\bf s}_r{\bf s}_r^H{\bf D})\right)+\log\left({\rm det}({\bf C}({\bf Q}_c)^{-1})\right) \nonumber\\
{\rm s.t.} & &\biggl\{ \log({\rm det}\left({\bf R}_c+{\bar {\bf H}}_c({\bf I}_L \otimes {\bf Q}_c){\bar {\bf H}}_c^H\right))+\log({\rm det}({\bf R}_c^{-1}))+ \biggr. \nonumber\\
\biggl. 
& & \log({\rm det}({\bf R}_d)) +\log({\rm det}\left({\bf R}_d+{\bar {\bf H}}_d({\bf I}_L \otimes {\bf Q}_c){\bar {\bf H}}_d^H\right)^{-1}) \biggr\}\geq \log(2) {\tilde r}_m. 
\end{eqnarray}
Introducing the auxiliary matrices ${\bf X}\succeq 0, {\bar {\bf Y}} \succeq 0$, and ${\bf Z} \succeq 0$ and replacing the terms of the form $\log ({\rm det}({\bf A}^{-1}))$ in a same way as in (\ref{eqn25}), 
we can reformulate (\ref{eqn55}) as the following optimization problem. 
\begin{eqnarray}
\label{eqn56}
\max_{{\bf s}_r, {\bf Q}_c\succeq 0, {\bf X}\succeq 0, {\bar {\bf Y}} \succeq 0, {\bf Z} \succeq 0} & & \log\left({\rm det}({\bf C}({\bf Q}_c)+{\bf D}^H{\bf s}_r{\bf s}_r^H{\bf D})\right)+\log({\rm det}({\bf X}))-{\rm tr}\left( {\bf X}{\bf C}({\bf Q}_c)\right)+LN \nonumber\\
{\rm s.t.} & &\biggl\{ \log({\rm det}\left({\bar {\bf H}}_c({\bf s}_r{\bf s}_r^H+({\bf I}_L \otimes {\bf Q}_c)){\bar {\bf H}}_c^H+\sigma_c^2{\bf I}_{LM}\right))+\log({\rm det}({\bar {\bf Y}})) \biggr. \nonumber\\
\biggl. 
& & -{\rm tr}\left( {\bar {\bf Y}}\left({\bar {\bf H}}_d({\bf s}_r{\bf s}_r^H+({\bf I}_L \otimes {\bf Q}_c)){\bar {\bf H}}_d^H+\sigma_r^2{\bf I}_{LN}\right)\right)+LN \biggr. \nonumber\\
\biggl. 
& & +\log({\rm det}\left({\bar {\bf H}}_d{\bf s}_r{\bf s}_r^H{\bar {\bf H}}_d^H+\sigma_r^2{\bf I}_{LN}\right))+\log({\rm det}({\bf Z})) \biggr. \nonumber\\
& & -{\rm tr}\left( {\bf Z}\left({\bar {\bf H}}_c{\bf s}_r{\bf s}_r^H{\bar {\bf H}}_c^H+\sigma_c^2{\bf I}_{LM}\right)\right)+LM \biggr\}\geq {\hat r}_m,
\end{eqnarray}
where ${\hat r}_m={\tilde r}_m \log(2)$.  We now introduce a new matrix variable ${\bar {\bf S}}_r={\bf s}_r{\bf s}_r^H \succeq 0$ and relax rank-one constraint of ${\bar {\bf S}}_r$. With the relaxation, we can solve the optimization problem (\ref{eqn56}) using alternating optimization approach. In particular, for a given ${\bar {\bf S}}_r$ and ${\bf Q}_c$, the solutions of ${\bf X}$, ${\bar {\bf Y}}$, and ${\bf Z}$ are expressed as
\begin{eqnarray}
\label{eqn57}
& & {\bf X}=[{\bf C}({\bf Q}_c)]^{-1},~{\bar {\bf Y}}=\left({\bar {\bf H}}_d( {\bar {\bf S}}_r+({\bf I}_L \otimes {\bf Q}_c)){\bar {\bf H}}_d^H+\sigma_r^2{\bf I}_{LN}\right)^{-1}, \nonumber\\
& & {\bf Z}=\left({\bar {\bf H}}_c{\bar {\bf S}}_r{\bar {\bf H}}_c^H+\sigma_c^2{\bf I}_{LM}\right)^{-1}.
\end{eqnarray}
On the other hand, for a given $\{ {\bf X}, {\bar {\bf Y}}, {\bf Z}\}$, the optimization over ${\bf S}_r$ and ${\bf Q}_c$ is the following SDR problem:
\begin{eqnarray}
\label{eqn58}
\max_{{\bar {\bf S}}_r\succeq 0, {\bf Q}_c\succeq 0} & & \log\left({\rm det}({\bf C}({\bf Q}_c)+{\bf D}^H{\bar {\bf S}}_r{\bf D})\right)-{\rm tr}\left( {\bf X}{\bf C}({\bf Q}_c)\right) \nonumber\\
{\rm s.t.} & &\biggl\{ \log({\rm det}\left({\bar {\bf H}}_c( {\bar {\bf S}}_r+({\bf I}_L \otimes {\bf Q}_c)){\bar {\bf H}}_c^H+\sigma_c^2{\bf I}_{LM}\right))+\log({\rm det}({\bar {\bf Y}})) \biggr. \nonumber\\
\biggl. 
& & -{\rm tr}\left( {\bar {\bf Y}}\left({\bar {\bf H}}_d( {\bar {\bf S}}_r+({\bf I}_L \otimes {\bf Q}_c)){\bar {\bf H}}_d^H+\sigma_r^2{\bf I}_{LN}\right)\right)+LN \biggr. \nonumber\\
\biggl. 
& & +\log({\rm det}\left({\bar {\bf H}}_d {\bar {\bf S}}_r{\bar {\bf H}}_d^H+\sigma_r^2{\bf I}_{LN}\right))+\log({\rm det}({\bf Z})) \biggr. \nonumber\\
& & -{\rm tr}\left( {\bf Z}\left({\bar {\bf H}}_c {\bar {\bf S}}_r{\bar {\bf H}}_c^H+\sigma_c^2{\bf I}_{LM}\right)\right)+LM \biggr\}\geq {\hat r}_m. 
\end{eqnarray}
The relaxed optimization problem (\ref{eqn58}) can be solved using iterative approach as in the case of ${\bf Algorithm 1}$.  In each iteration, the convex optimization problem w.r.t. ${\bar {\bf S}}_r$  and ${\bf Q}_c$  is solved by keeping the auxiliary variables, $\{ {\bf X}, {\bar {\bf Y}}, {\bf Z}\}$,  fixed, and then the  convex optimization problem w.r.t. auxiliary variables is solved by keeping  ${\bar {\bf S}}_r$  and ${\bf Q}_c$  fixed. As such, this method is an AO method and its convergence can be proven with a similar way as in  ${\bf Algorithm 1}$. After convergence, if the optimum  ${\bar {\bf S}}_r$  is rank-one, then it will also be the optimum solution of the original problem  (\ref{eqn56}). Otherwise, randomization techniques can be applied to approximate rank-one solutions from ${\bar {\bf S}}_r$ \cite{ChaliseVandendorpe}.\\
\hspace*{0.3cm} In contrast to the optimization problem in  orthogonal (non-overlapping) case, the SINR at the RR as well as secrecy information rate depend on the actual values of  ${\bar {\bf S}}_r$ and ${\bf Q}_c$. This suggests that the overlapping case provides additional degrees of freedom to maximize the SINR and satisfy the target secrecy rate. Consequently, although overlapping case causes interference (to RR from communication signal) and (to CR from radar waveform), better performance can be obtained through joint optimization of ${\bar {\bf S}}_r$ and ${\bf Q}_c$. Our numerical simulations of next section  also justify this argument. However, this improvement in performance is achieved with a increased complexity. The reason is that, in contrast to the SDP in non-overlapping case, problem (\ref{eqn58}) has an additional matrix variable, ${\bar{\bf S}}_r$, of size $LN_t \times LN_t$. This means that the complexity of the corresponding algorithm increases polynomially in $LN_t$ (in contrast to only $N_t$ in non-overlapping case). It is worthwhile to comment that, in the non-overlapping case,  the secrecy rate does not depend on the radar waveform, ${\bf s}_r$,  (see (\ref{eqn13N})) , whereas the radar SINR does not depend on the transmit covariance matrix, ${\bf Q}_c$,   of information signals (see (\ref{eqn19})). However, in the overlapping case, the secrecy rate as well as the radar SINR depend on both radar waveform and transmit covariance matrix  (see (\ref{eqn45}) and  (\ref{eqn49}) ).

\section{Numerical Results}
\label{sec4}
In this section, we first simulate the performance of the proposed SDP-based ({\bf Algorithm 1}) and semi-analytical  ({\bf Algorithm 2}) approaches for the case in which the radar waveforms and information signals occupy orthogonal (non-overlapping) resources. We then demonstrate the performance of the proposed SDR-based algorithm for the case where the radar waveforms and information signals occupy the same sets of resources. Throughout all simulations, we consider that the transmitter, RR, and CR employ uniform linear array (ULA) with half-wavelength inter-element spacing. We illustrate the performance of the proposed methods by choosing $\theta_t=40^{\circ}$,  $\theta_r=42^{\circ}$, $M=N$, and $P_T=30$ W. The hypothesized target position's location is such that  $\theta_{t,0}=30^{\circ}$,  $\theta_{r,0}=32^{\circ}$. The average signal-to-noise ratios (SNR) associated with the direct path and surveillance channels are set to 20 dB and 10 dB, respectively, whereas that associated with the transmitter-CR channel is set to $0$ dB. The elements of the transmitter-CR channel are assumed to be zero-mean complex i.i.d. Gaussian. All results correspond to averaging over $100$ simulation runs, and we take $N_{it}=100$ and $\epsilon=0.01$.  

\begin{figure}[htb!]
\includegraphics[width=.45\linewidth]{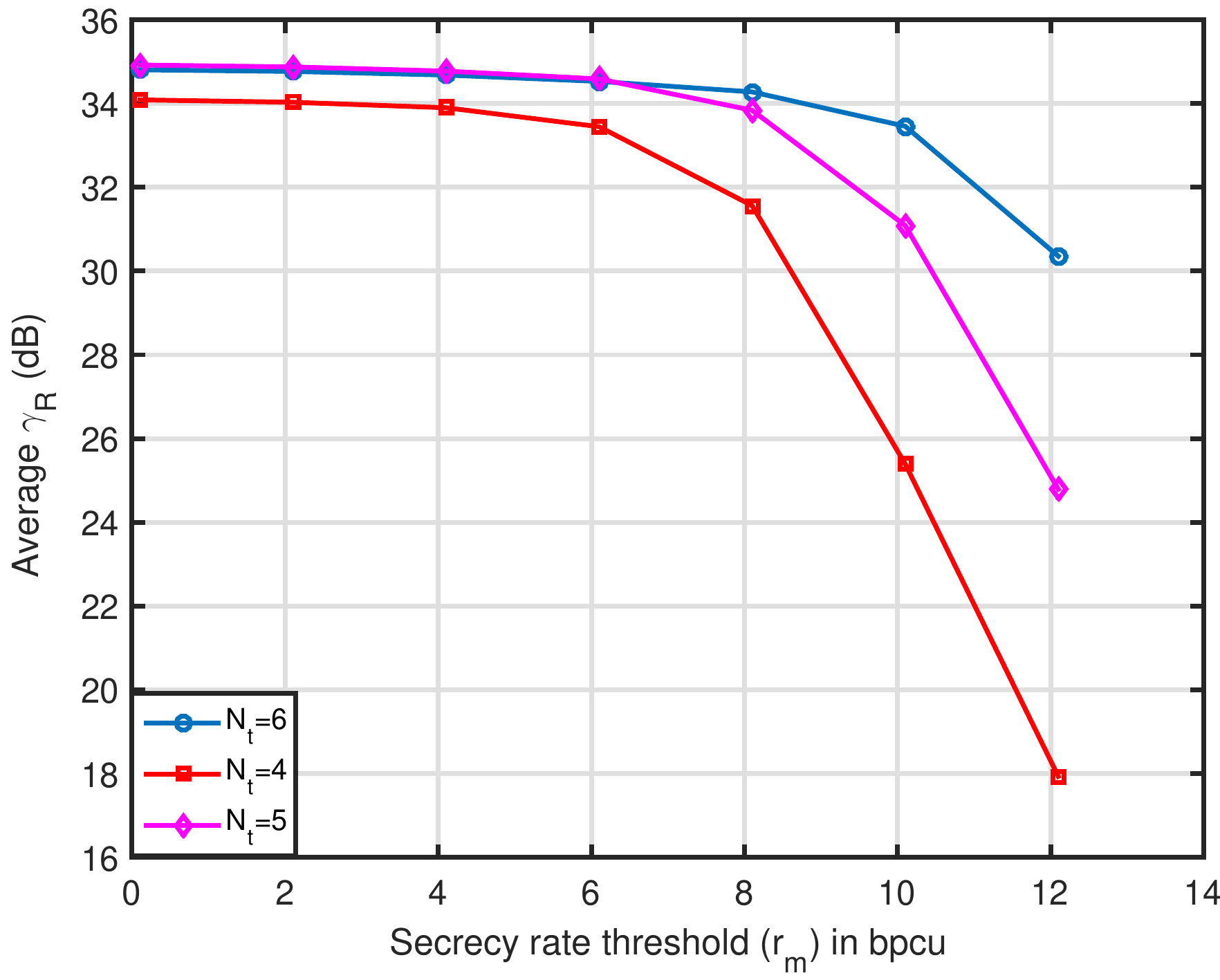}\hfill
\includegraphics[width=.45\linewidth]{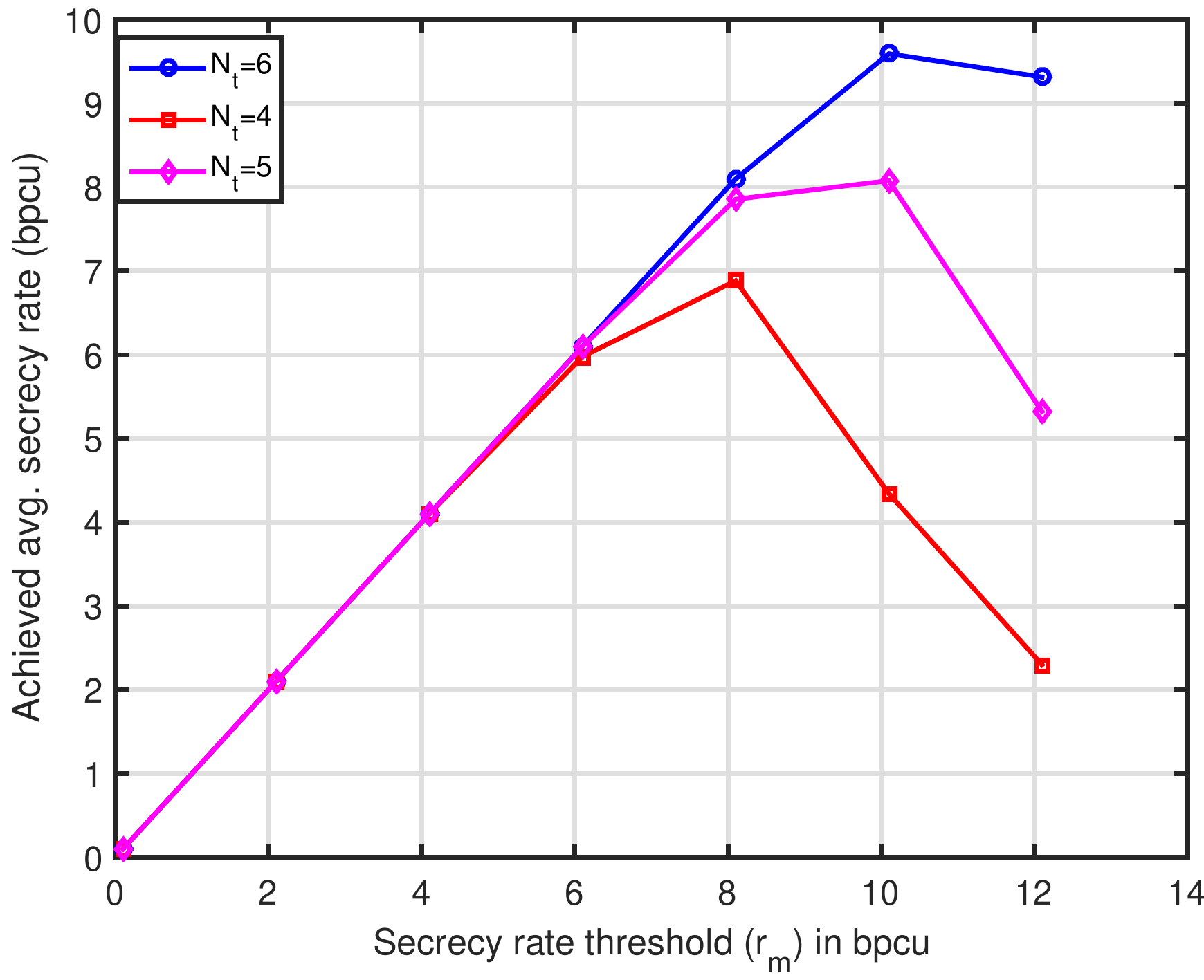}
\caption{Performance of {\bf Algorithm 1} for different $N_t$:  Left - Secrecy rate threshold versus SINR at the RR, Right -  Secrecy rate threshold vs achieved average secrecy rate}
\label{simRes1}
\end{figure}
The performance of the SDP-based method ({\bf Algorithm 1}) is shown in Fig. \ref{simRes1} for different values of $N_t$, where we set $M=N=4$ and $L=10$.  It can be observed from  Fig. \ref{simRes1}-Left that the average SINR at the RR  starts to decrease when the secrecy rate threshold, $r_m$,  increases. This decrease is, however,  significant after $r_m$ reaches a certain value and for smaller values of $N_t$.  Fig. \ref{simRes1}-Right plots the achieved average secrecy rate as a function of $r_m$. It can be observed from this figure that the achieved average secrecy rate is same as $r_m$ for its smaller values. At larger values of $r_m$, the achieved average secrecy rate drops significantly. This is due to the fact that the feasibility of the SDP decreases when $r_m$  increases. More specifically, the number of channels, for which the secrecy rate threshold cannot be met, increases as $r_m$ increases. 

\begin{figure}[htb!]
\includegraphics[width=.45\linewidth]{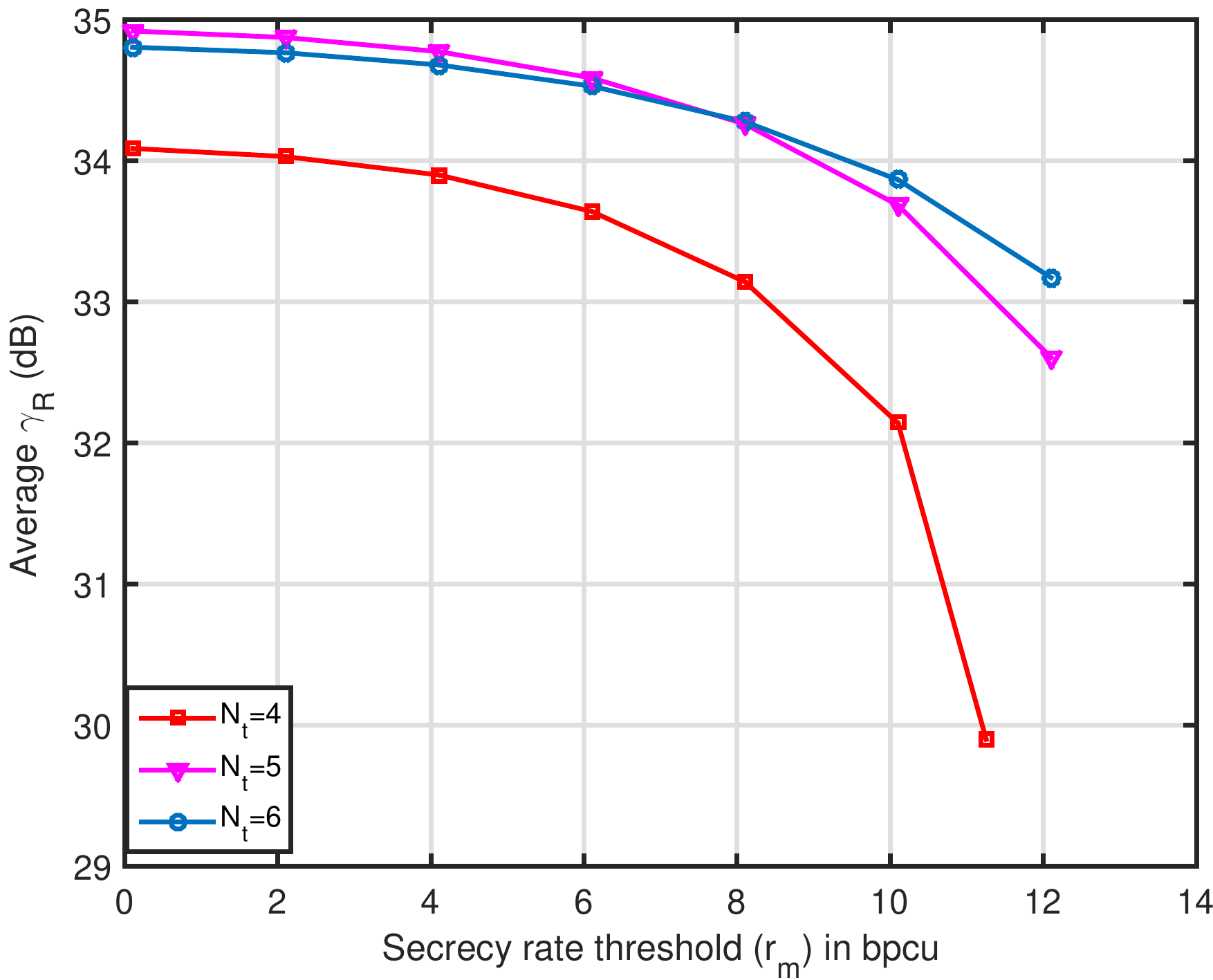}\hfill
\includegraphics[width=.45\linewidth]{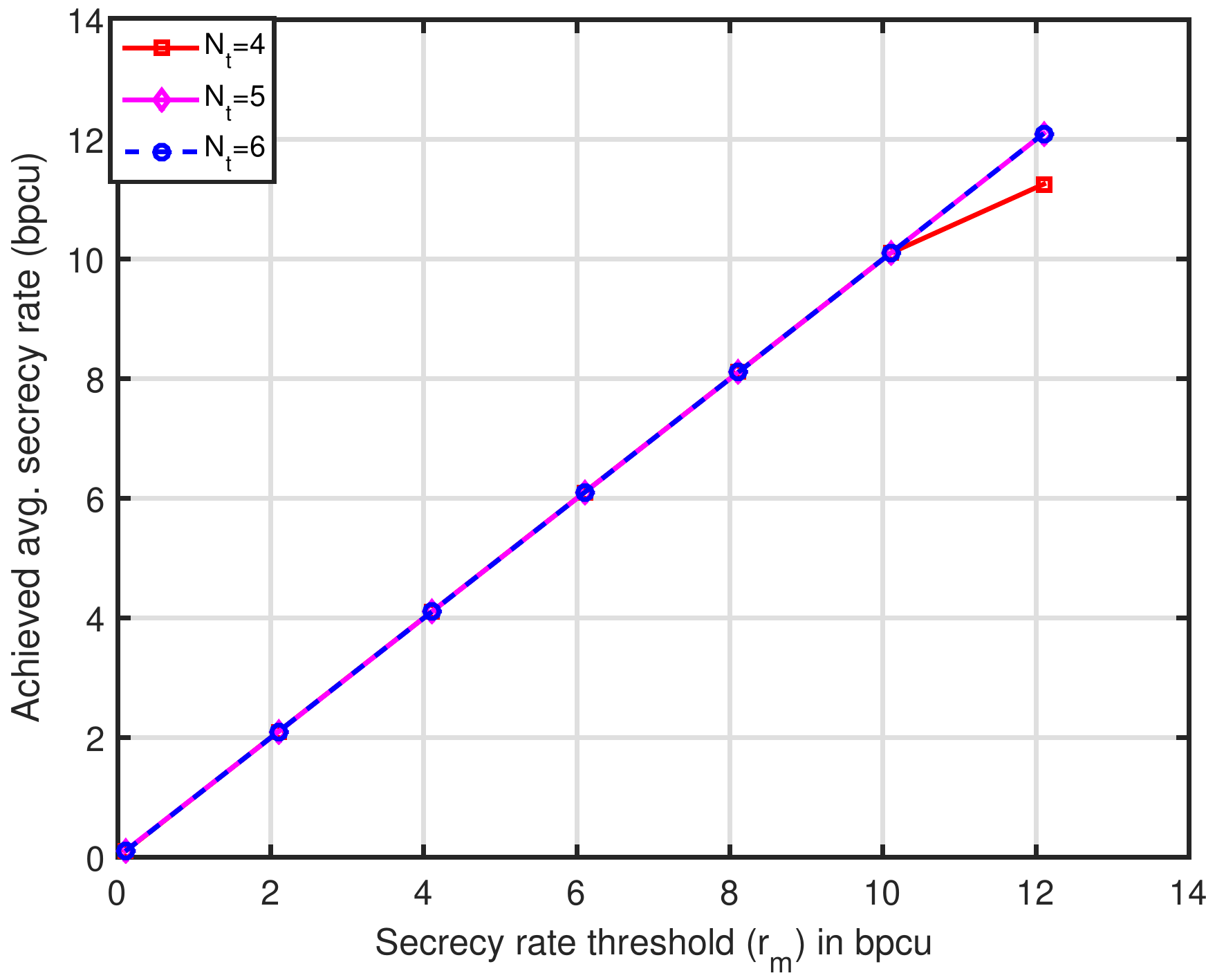}
\caption{Performance of {\bf Algorithm 2} for different $N_t$:  Left - Secrecy rate threshold versus SINR at the RR, Right -  Secrecy rate threshold vs achieved average secrecy rate}
\label{simRes2}
\end{figure}

Fig. \ref{simRes2} demonstrates the performance of the proposed semi-analytical approach ({\bf Algorithm 2}) for the same set of parameters as in Fig. \ref{simRes1}. Although the achieved SINR at the RR decreases when $r_m$ increases, the decrease in SINR (especially at larger values of $r_m$) is not rapid as in the SDP-based method. Moreover,  Fig. \ref{simRes2}-Right shows that the achieved average secrecy rate is same as $r_m$ except at its largest value when $N_t=4$. This is due to the fact that,  when $r_m$  increases, the feasibility of the semi-analytical approach decreases at a much smaller rate than the SDP-based method. In a nutshell, by comparing   Fig. \ref{simRes2}  and Fig. \ref{simRes1}, it is clear that the semi-analytical approach provides much better performance than the SDP based approach. This is a significant advantage since the computational complexity of the proposed semi-analytical approach is much less than that of the SDP-based method.

\begin{figure}[htb!]
\includegraphics[width=.45\linewidth]{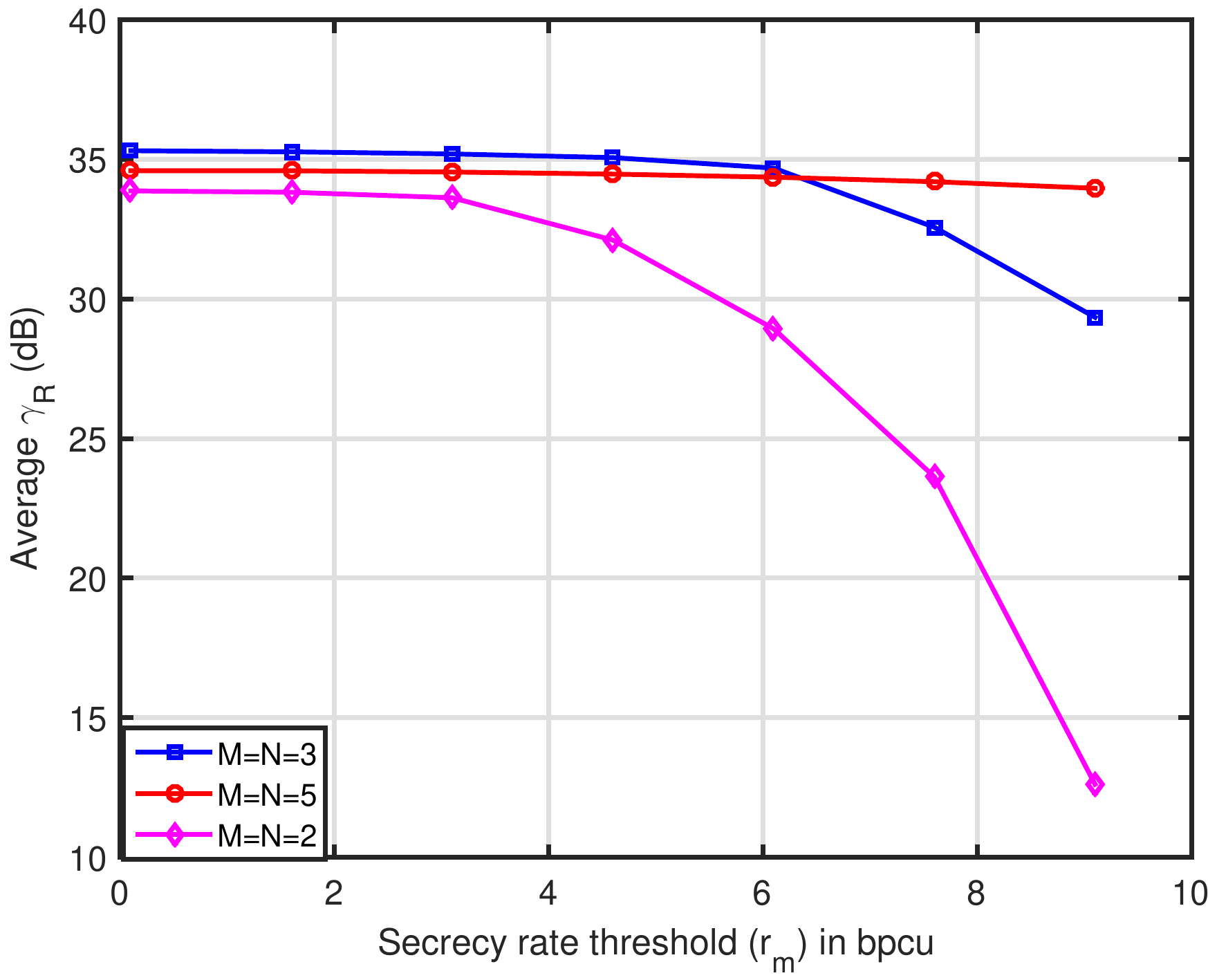}\hfill
\includegraphics[width=.45\linewidth]{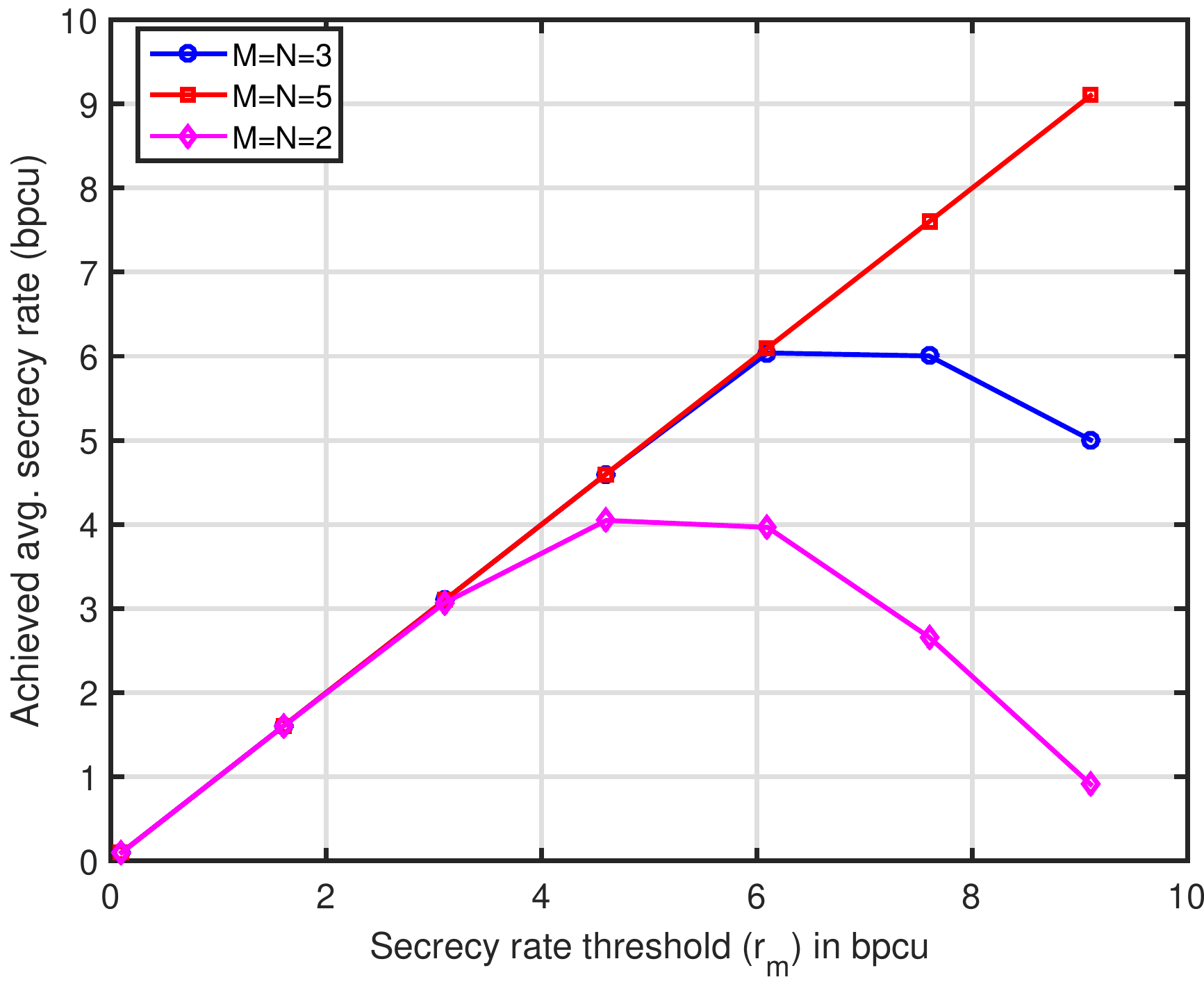}
\caption{Performance of {\bf Algorithm 1} for different $M=N$ and fixed $N_t$:  Left - Secrecy rate threshold versus SINR at the RR, Right -  Secrecy rate threshold vs achieved average secrecy rate}
\label{simRes3}
\end{figure}

The performance of the SDP-based method ({\bf Algorithm 1}) is shown in Fig. \ref{simRes3} for different values of $M=N$ and the fixed values of $N_t=5$ and $L=10$. As in  Fig. \ref{simRes1}, the performance of the proposed method drops rapidly at larger values of  $r_m$ and smaller values of $N=M$. This can be attributed to the fact that the feasibility of the SDP-based method decreases at a faster rate when $r_m$ takes larger values. However, as shown in  Fig. \ref{simRes4} (with the same setting as in  Fig. \ref{simRes3}) , the performance drop caused by higher infeasibility of the SDP-based method can be minimized by using the proposed semi-analytical approach. 

\begin{figure}[htb!]
\includegraphics[width=.45\linewidth]{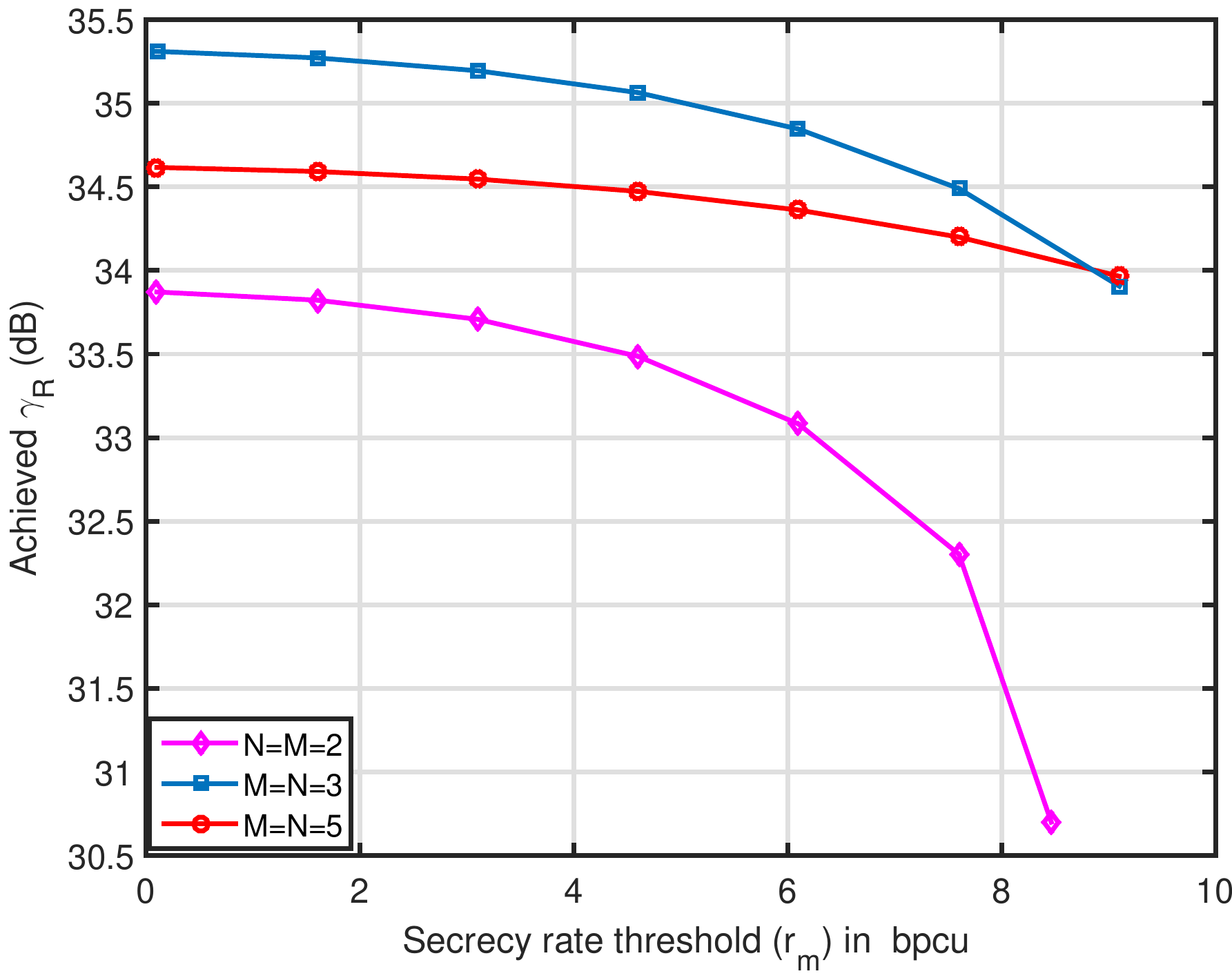}\hfill
\includegraphics[width=.45\linewidth]{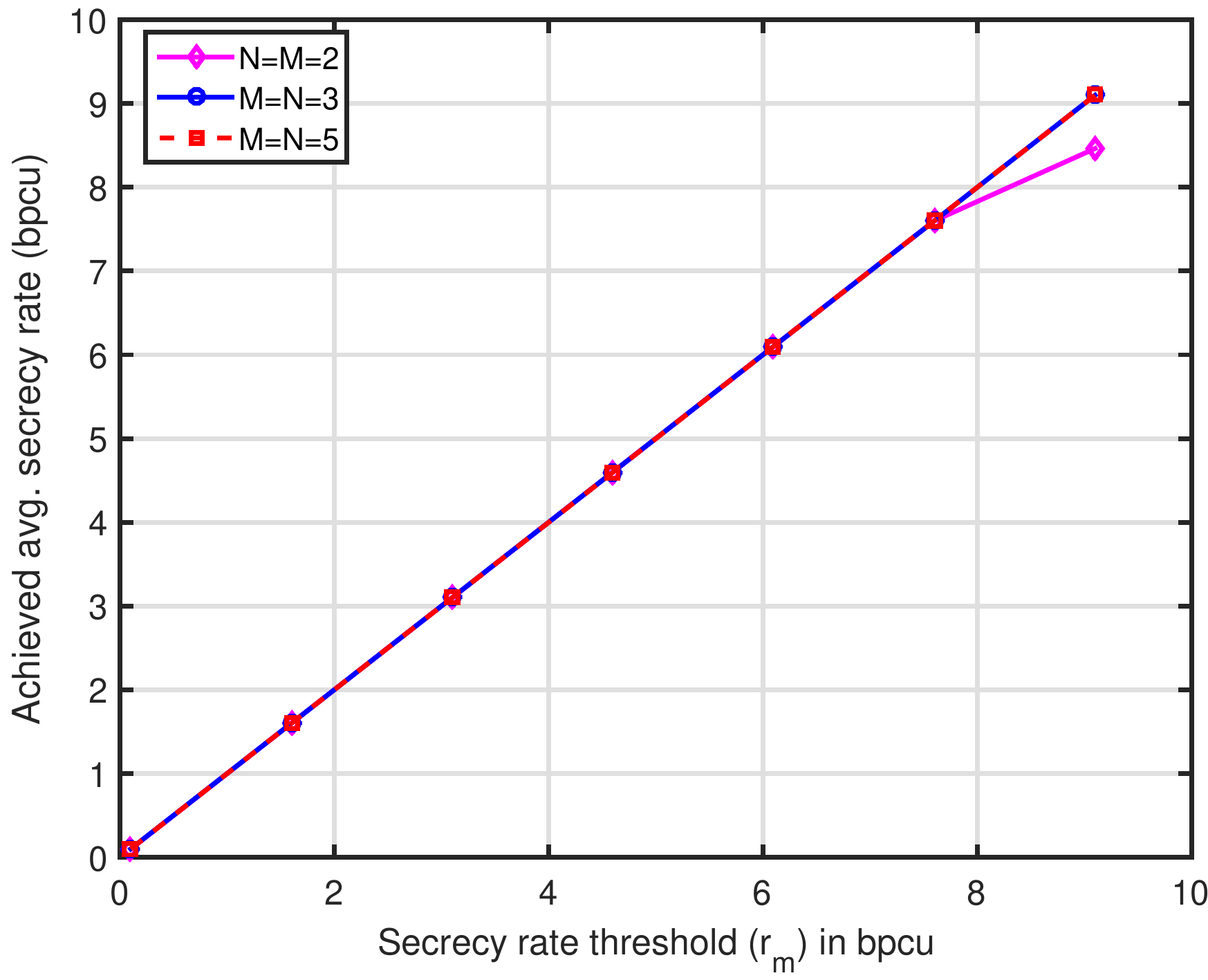}
\caption{Performance of {\bf Algorithm 2}  for different $M=N$ and fixed $N_t$:  Left - Secrecy rate threshold versus SINR at the RR, Right -  Secrecy rate threshold vs achieved average secrecy rate}
\label{simRes4}
\end{figure}

In Fig. \ref{simRes5}, we compare the performance of the proposed method for the case wherein radar waveforms and information signals occupy same sets of resources (overlapping case). We take $N_t=2$, $L=3$,  and change $M=N$ in this figure. It can be observed from this figure that the achieved SINR decreases with the larger values of  ${\tilde r}_m$ and the smaller value of $N=M$. Moreover, for the smaller value of $N=M$, the achieved average secrecy rate is smaller than the secrecy rate threshold, which suggests that the feasibility of SDR decreases with decreasing $N=M$.  
\begin{figure}[htb!]
\includegraphics[width=.45\linewidth]{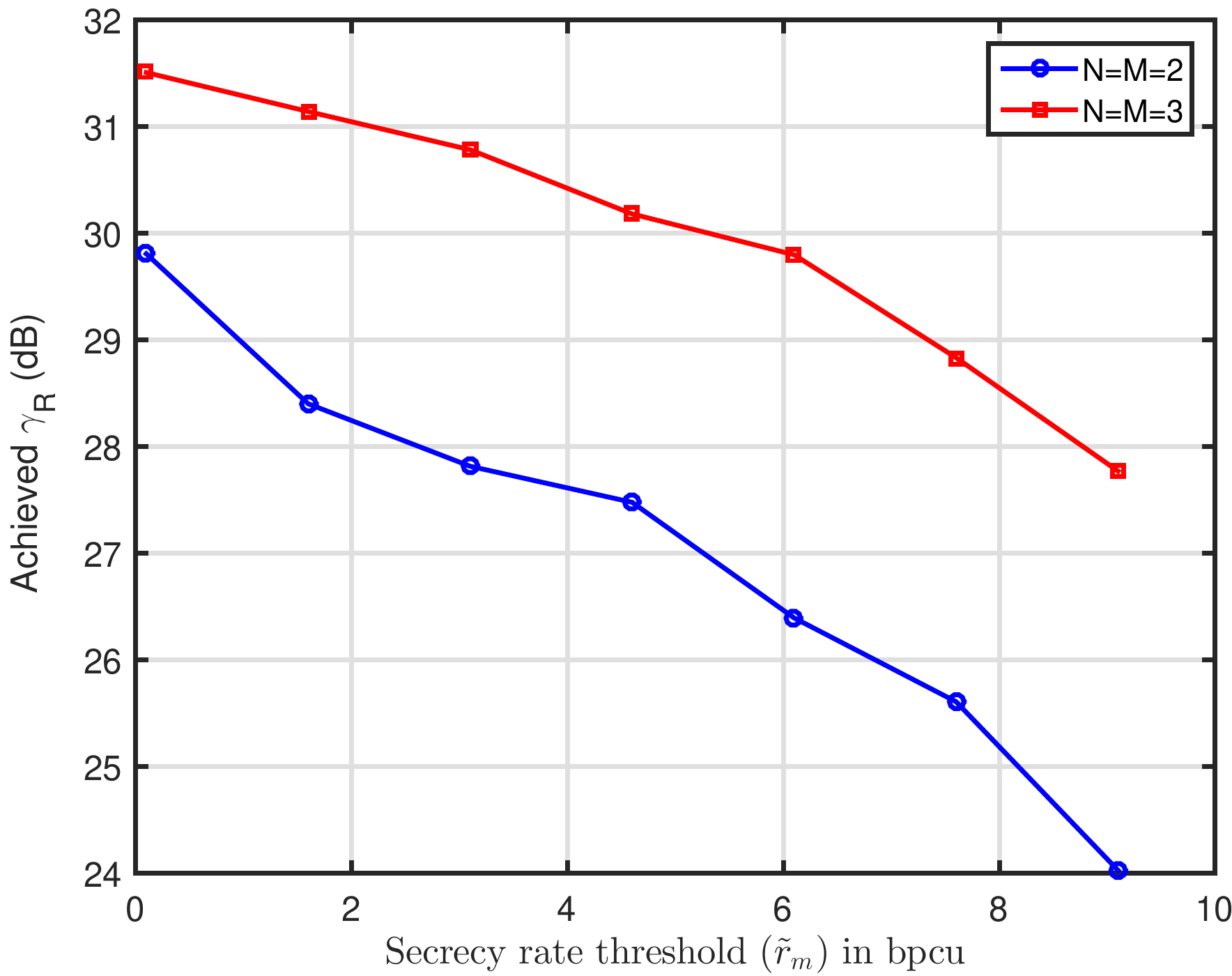}\hfill
\includegraphics[width=.45\linewidth]{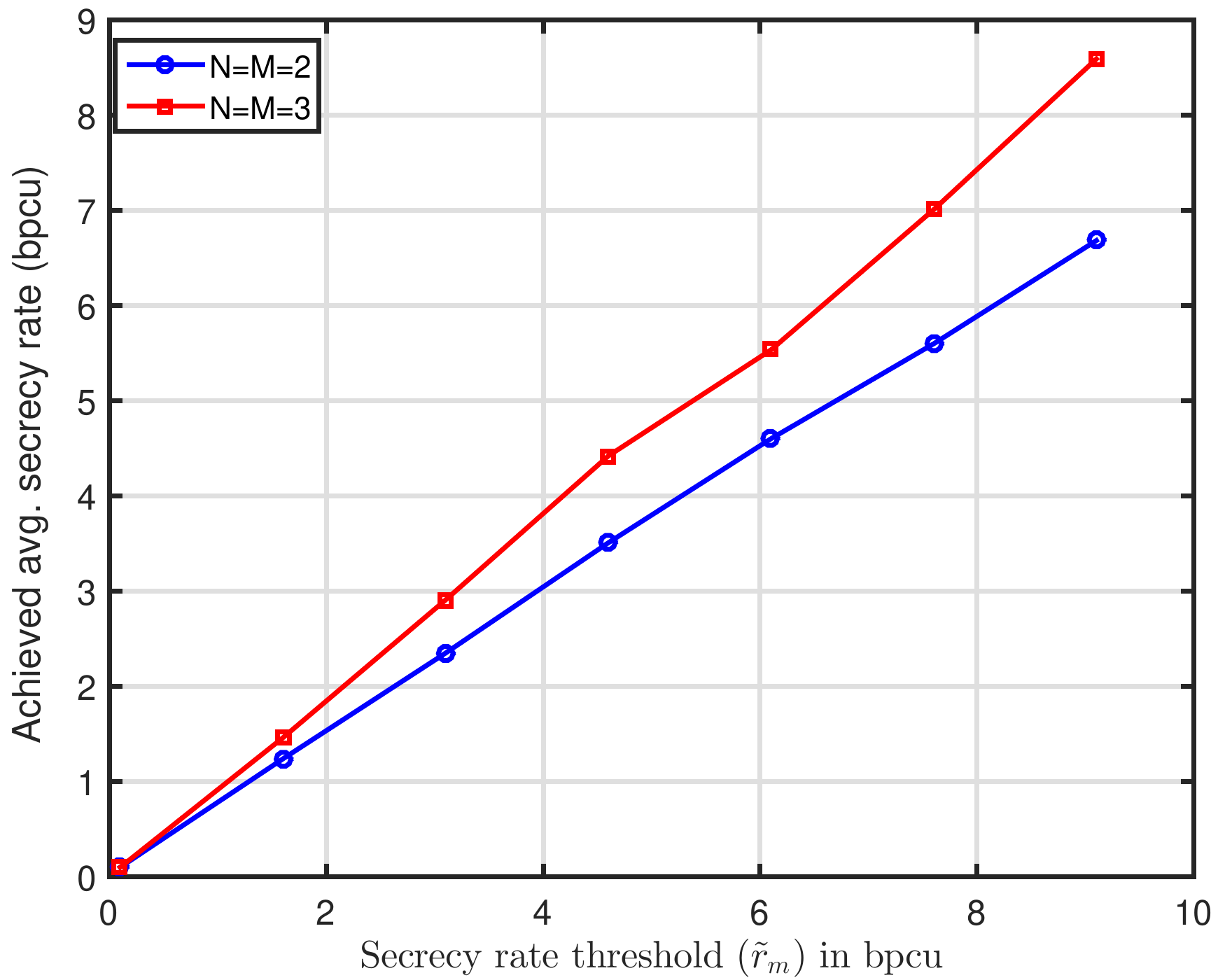}
\caption{Performance of the overlapping case:  Left - Secrecy rate threshold versus SINR at the RR, Right -  Secrecy rate threshold vs achieved average secrecy rate}
\label{simRes5}
\end{figure}

The performance between the overlapping and non-overlapping ({\bf Algorithm 2}) cases is shown in Fig. \ref{simRes6} for $N_t=2$, $L=3$, and $M=N=3$. It can be observed from this figure (Left side) that the achieved SINR drops significantly in the non-overlapping case for larger values of secrecy rate threshold. In particular, for the largest threshold value of 9.1 bpcu, {\bf Algorithm 2} turns out to be infeasible for all channel realizations. On the other hand, the SINR in the overlapping case remains relatively stable for all values of the secrecy rate threshold. Moreover, Fig. \ref{simRes6} (Right side) shows that the achieved average secrecy rate in overlapping case is almost close to the threshold, whereas the average secrecy rate in the non-overlapping case drops to zero at the largest threshold value of this setting. This shows that when the radar waveform and covariance matrix of information signals are jointly optimized, the overlapping case outperforms the non-overlapping case. 
\begin{figure}[htb!]
\includegraphics[width=.45\linewidth]{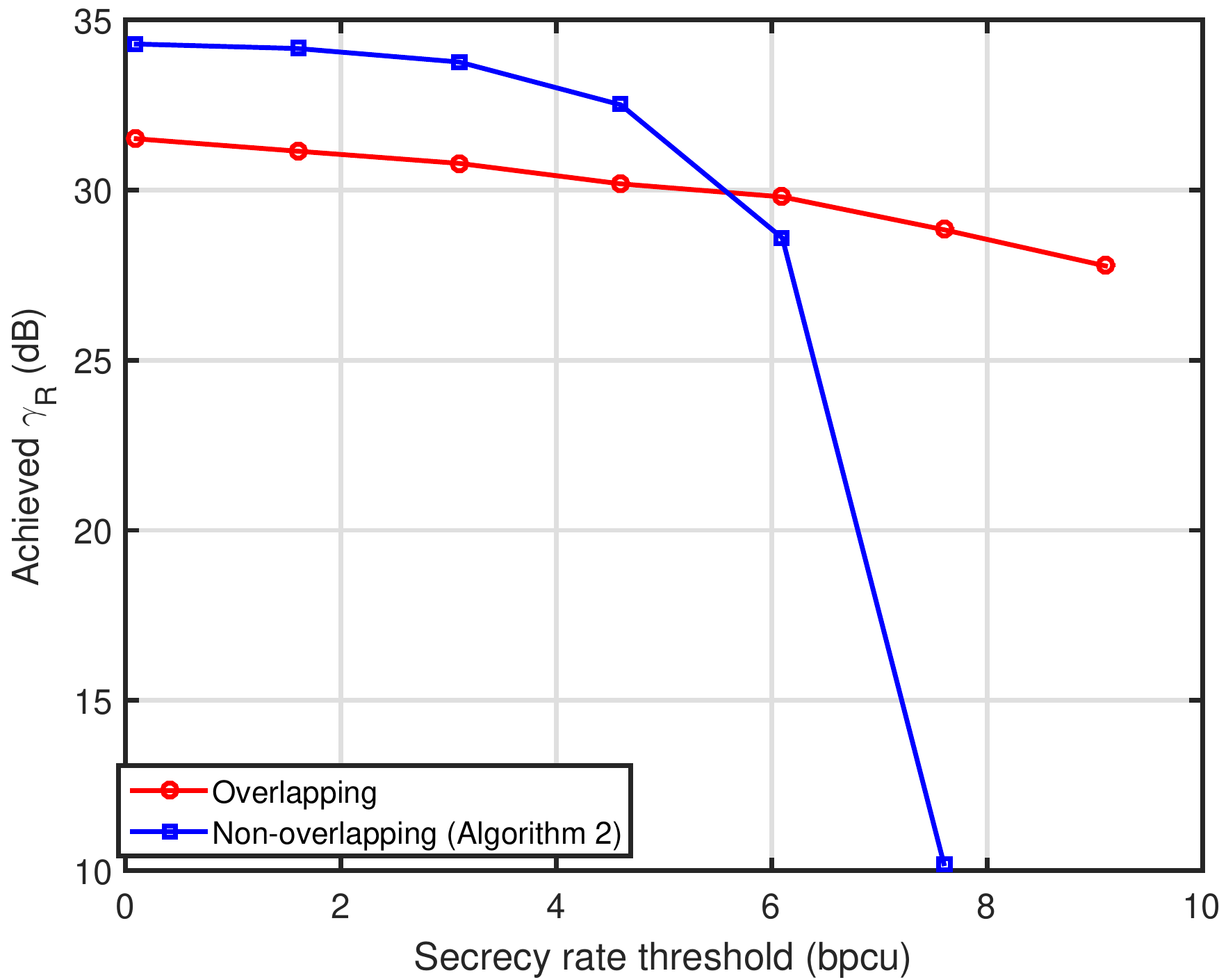}\hfill
\includegraphics[width=.45\linewidth]{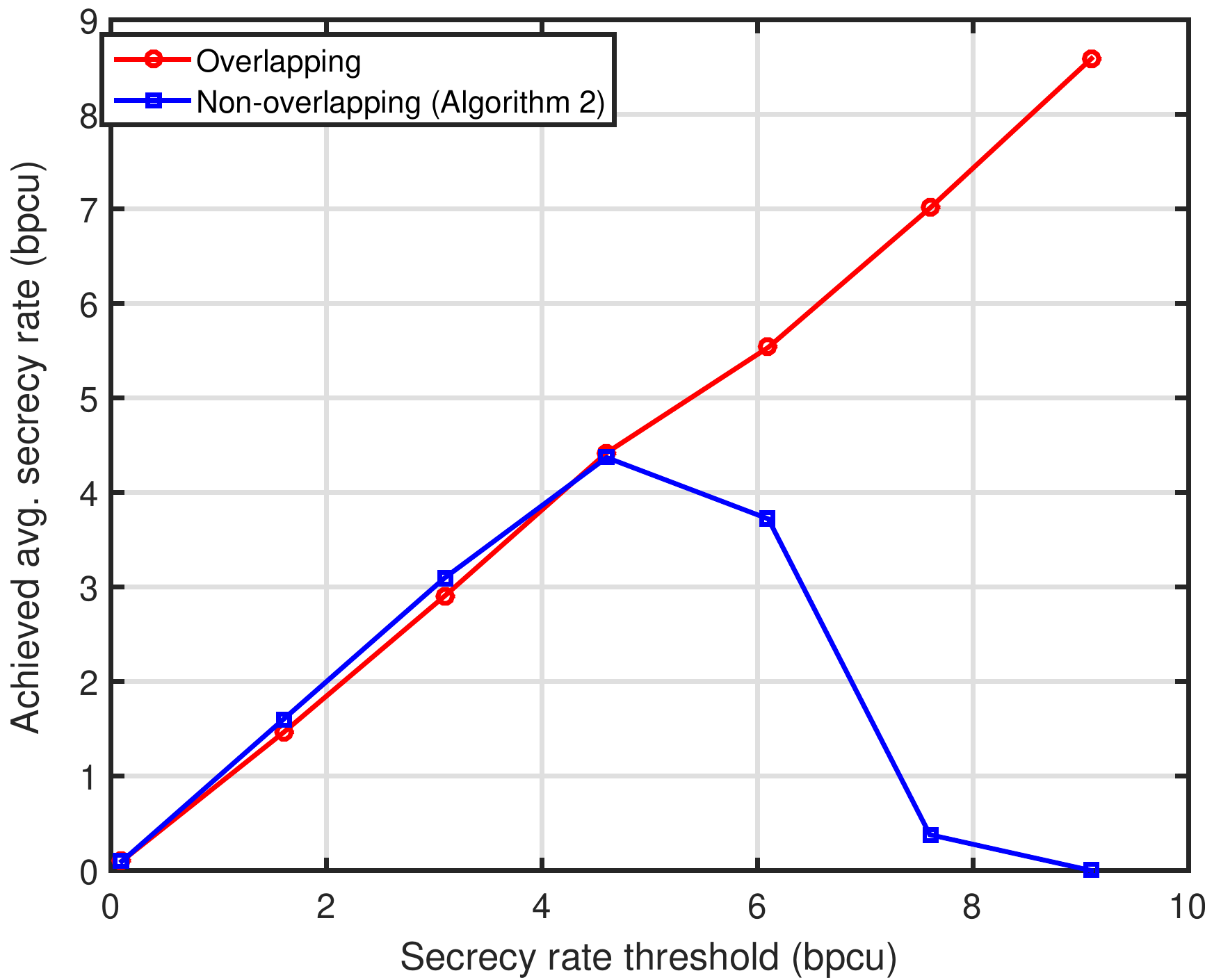}
\caption{Comparison between overlapping and non-overlapping cases:  Left - Secrecy rate threshold versus SINR at the RR, Right -  Secrecy rate threshold vs achieved average secrecy rate}
\label{simRes6}
\end{figure}

 \section{Conclusions}
 \label{sec5}
 In this paper, we analyzed the performance tradeoff between  radar and communications in a unified system consisting of  a transmitter, a passive radar receiver, and a communication receiver, all equipped with multiple antennas. The tradeoff was characterized by obtaining the boundaries of the signal-to-interference-and-noise ratio (SINR) for the radar receiver versus information secrecy rate region. To this end, optimization problems, with the objective of maximizing the SINR at the radar receiver while ensuring that the information secrecy rate is above a certain threshold, are formulated when radar and information signals use both non-overlapping and overlapping sets of resources. In both cases, iterative alternating optimization methods are proposed for optimizing the  radar waveforms and transmit covariance matrices of information signals. However, in the former case, in contrast to the  iterative approach that employs semi-definite programming (SDP), a computationally efficient semi-analytical approach was also proposed. Simulation results show that this approach provides significant performance gains over the SDP-based approach. In the latter (overlapping) case,  the optimization problem is non-tractable and challenging. However, it was reformulated as a semi-definite relaxation problem and solved iteratively under a framework of alternating optimization methods.

 \section*{Appendix : Proof of Proposition 1}
 The Lagrangian multiplier function  (\ref{eqn29}) can be expressed as
 \begin{eqnarray}
\label{eqn1Ap1}
{\mathcal L}({\bf Q}_c, \lambda)&=&{\rm tr}({\bf Q}_c)+\lambda {\rm tr}\left( {\bf H}_d^H{\bf Y}{\bf H}_d{\bf Q}_c\right)-\lambda  \log\left( {\rm det}\left({\bf H}_c{\bf Q}_c {\bf H}_c^H+\sigma_c^2{\bf I}_{M}\right)\right)+\nonumber\\
& &\lambda\left({\bar r}_m-\log\left({\rm det}({\bf Y})\right)-{\bar N}+\sigma_r^2{\rm tr}({\bf Y})\right).
\end{eqnarray}
 For a given $\lambda$ and ${\bf Y}\succeq 0$,  ${\mathcal L}({\bf Q}_c, \lambda)$ can be minimized from 
\begin{eqnarray}
\label{eqn2Ap1}
\min_{{\bf Q}_c\succeq 0} {\rm tr}\left( \left({\bf I}_{N_t}+\lambda {\bf H}_d^H{\bf Y}{\bf H}_d\right){\bf Q}_c \right)-\lambda  \log\left( {\rm det}\left({\bf H}_c{\bf Q}_c {\bf H}_c^H+\sigma_c^2{\bf I}_M\right)\right).
\end{eqnarray}
 Define ${\bf P}\triangleq  {\bf I}_{N_t}+\lambda {\bf H}_d^H{\bf Y}{\bf H}_d={\bf P}^{\frac{1}{2}}{\bf P}^{\frac{H}{2}}$, ${\bf Q}_c={\bf Q}_c^{\frac{1}{2}}{\bf Q}_c^{\frac{H}{2}}$, and ${\tilde {\bf Q}}_c^{\frac{1}{2}}={\bf P}^{\frac{H}{2}}{\bf Q}_c^{\frac{1}{2}}$. Substituting these relations (including ${\bf Q}_c^{\frac{1}{2}}={\bf P}^{-\frac{H}{2}}{\tilde {\bf Q}}_c^{\frac{1}{2}}$) into (\ref{eqn2Ap1}), it can be expressed in terms of ${\tilde {\bf Q}}_c$ as
 \begin{eqnarray}
\label{eqn3Ap3}
 \min_{{\tilde {\bf Q}}_c\succeq 0} {\rm tr}\left( {\tilde {\bf Q}}_c \right)-\lambda  \log\left( {\rm det}\left({\bf H}_c {\bf P}^{-\frac{H}{2}}{\tilde {\bf Q}}_c {\bf P}^{-\frac{1}{2}} {\bf H}_c^H+\sigma_c^2{\bf I}_M\right)\right).
\end{eqnarray}
Clearly, the minimum of (\ref{eqn3Ap3}) is obtained when  ${\rm det}\left({\bf H}_c {\bf P}^{-\frac{H}{2}}{\tilde {\bf Q}}_c {\bf P}^{-\frac{1}{2}} {\bf H}_c^H+\sigma_c^2{\bf I}_M\right)$ is maximized, which happens when Hadamard inequality \cite{HornJohnson} is satisfied with equality. This implies that the optimum ${\tilde {\bf Q}}_c$ will be such that ${\bf H}_c {\bf P}^{-\frac{H}{2}}{\tilde {\bf Q}}_c {\bf P}^{-\frac{1}{2}} {\bf H}_c^H+\sigma_c^2{\bf I}_M$ turns to a diagonal matrix. To this end, 
let the singular value decomposition (SVD) of  ${\bf H}_c {\bf P}^{-\frac{H}{2}}$ be given by  ${\bf H}_c {\bf P}^{-\frac{H}{2}}={\bf U}{\boldsymbol \Sigma}{\bf V}^H$, where ${\bf U}$ and ${\bf V}$ are $M\times M$ and $N_t \times N_t$  unitary matrices and 
 ${\boldsymbol \Sigma}$ is a diagonal matrix of elements $\{ d_i \}_{i=1}^{r}$, where $d_i>0$ and $r=\min(M, N_t)$. Substituting the SVD of  ${\bar {\bf H}}_c {\bf P}^{-\frac{H}{2}}$  into (\ref{eqn3Ap3}), we get
 \begin{eqnarray}
\label{eqn4Ap3}
 \min_{{\tilde {\bf Q}}_c\succeq 0} {\rm tr}\left( {\tilde {\bf Q}}_c \right)-\lambda  \log\left( {\rm det}\left({\bf U}{\boldsymbol \Sigma}{\bf V}^H{\tilde {\bf Q}}_c {\bf V} {\boldsymbol \Sigma}^T{\bf U}^H+\sigma_c^2{\bf I}_M\right)\right),
\end{eqnarray}
which can be simplified to
 \begin{eqnarray}
\label{eqn5Ap3}
  \min_{{\tilde {\bf Q}}_c\succeq 0} {\rm tr}\left( {\tilde {\bf Q}}_c \right)-\lambda  \log\left( {\rm det}\left({\boldsymbol \Sigma}{\bf V}^H{\tilde {\bf Q}}_c {\bf V} {\boldsymbol \Sigma}^T+\sigma_c^2{\bf I}_M\right)\right).
\end{eqnarray}
For ${\bf V}^H {\tilde {\bf Q}}_c  {\bf V}$ to be diagonal, $ {\tilde {\bf Q}}_c $ must be ${\bf V}{\boldsymbol \Lambda} {\bf V}^H$, where ${\boldsymbol \Lambda}={\rm diag}(\mu_1, \cdots, \mu_{N_t})$. Substitution of such  $ {\tilde {\bf Q}}_c$  leads to the following equivalent optimization problem
\begin{eqnarray}
\label{eqn6Ap3}
  \min_{\{\mu_i\}_{i=1}^{r}} \sum_{i=1}^{r} \mu_i-\lambda \sum_{i=1}^{r}\log(\mu_i d_i^2+\sigma_c^2),
\end{eqnarray}
where $\{\mu_j\}_{j=r+1}^{N_t}=0$ can be chosen without any loss of generality. Solving (\ref{eqn6Ap3}) in terms of $\mu_i$, we get
\begin{eqnarray}
\mu_i=\left[ \lambda -\frac{\sigma_c^2}{d_i^2}\right]^{+},  \forall i, 
\end{eqnarray}
where $[x]^{+}=\max(0, x)$. The proof of Proposition 1 is complete.

\small{ 
 
}
 
 \end{document}